\pdfoutput=1
\documentclass[twocolumn,a4paper,10pt,nopdfoutputerror,floatfix]{quantumarticle}
\usepackage[utf8]{inputenc}
\usepackage[english]{babel}
\usepackage[T1]{fontenc}
\usepackage{amsmath}
\usepackage{hyperref}
\usepackage[numbers,sort&compress]{natbib}
\usepackage{booktabs}
\usepackage{xtab,longtable}
\usepackage{amsfonts}
\usepackage{amsthm}
\usepackage{physics}
\usepackage{tikz}
\usetikzlibrary{quantikz}
\usepackage{adjustbox}
\newtheorem{thm}{Theorem}
\newtheorem{lem}[thm]{Lemma}

\newtheorem{prop}[thm]{Proposition}

\newtheorem{defn}[thm]{Definition}
\newtheorem{example}[thm]{Example}
\newtheorem{remark}[thm]{Remark}
\usepackage{xcolor}
\usepackage{graphicx}
\usepackage{dcolumn}
\usepackage{bm}
\usepackage{verbatim}
\usepackage{stmaryrd}

\newcommand\dsb[1]{\llbracket #1 \rrbracket}
\newcommand\Diag{\operatorname{diag}}

\begin{document}

\title{How Much Entanglement Does a Quantum Code Need?}

\author{Gaojun Luo}
  \orcid{0000-0003-1482-4783}
  \email{gaojun.luo@ntu.edu.sg}
  \affiliation{%
    School of Physical and Mathematical Sciences, Nanyang Technological University,
    21 Nanyang Link, Singapore 637371, Singapore
  }%
\author{Martianus Frederic Ezerman}%
  \orcid{0000-0002-5851-2717}
  \email{fredezerman@ntu.edu.sg\\fred@sandhiguna.com}
  \affiliation{%
    School of Physical and Mathematical Sciences, Nanyang Technological University,
    21 Nanyang Link, Singapore 637371, Singapore and Sandhiguna, Suite 707 Graha Pena, Batam, Kepulauan Riau, 29461, Indonesia
}%
\author{Markus Grassl}%
  \orcid{0000-0002-3720-5195}%
  \email{markus.grassl@ug.edu.pl}
  \affiliation{%
    International Centre for Theory of Quantum Technologies, University of Gdansk,
    80-309 Gda\'nsk, Poland
  }
\author{San Ling}%
  \orcid{0000-0002-1978-3557}
  \email{lingsan@ntu.edu.sg}
  \affiliation{%
    School of Physical and Mathematical Sciences, Nanyang Technological University,
    21 Nanyang Link, Singapore 637371, Singapore
}%

\date{6 September 2022}

\begin{abstract}
In the setting of entanglement-assisted quantum error-correcting codes (EAQECCs), the sender and the receiver have access to pre-shared entanglement. Such codes promise better information rates or improved error handling properties. Entanglement incurs costs and must be judiciously calibrated in designing quantum codes with good performance, relative to their deployment parameters. 

Revisiting known constructions, we devise tools from classical coding theory to better understand how the amount of entanglement can be varied. We present three new propagation rules and discuss how each of them affects the error handling. Tables listing the parameters of the best performing qubit and qutrit EAQECCs that we can explicitly construct are supplied for reference and comparison.
\end{abstract}

\keywords{algebraic codes, quantum codes, orthogonal codes, entanglement assisted}
\maketitle

\section{Introduction}\label{sec:intro}
In quantum communication, the goal is to send as much quantum
information as possible over a noisy quantum channel using a fixed
number of quantum bits (\textit{qubits}) or higher-dimensional systems (\textit{qudits}). One aims at optimizing the
transmission rate so that it approaches the channel capacity
asymptotically.  The communicating parties are assumed to be
physically separated, but they might have access to additional
resources, which may include access to classical communication
channels, pre-shared randomness, and pre-shared entanglement. We focus on the use of entanglement in the design of quantum error-correcting codes (QECCs) to boost either their communication
rates or error-control capabilities.

\subsection{Quantum Codes}\label{subsec:QC}

A general quantum error-correcting code for qubits that does not use
additional resources is a $K$-dimensional subspace of the complex
Hilbert space of $n$ qubits, which has dimension $2^n$. We use the
notation $\dsb{n,k}_2$ for a code $\mathcal{Q}$ of dimension $K=2^k$
and say that $\mathcal{Q}$ encodes $k$ logical qubits into $n$
physical qubits. The encoding operation consists of two steps.  First,
the sender appends $n-k$ ancilla qubits in a fixed state, typically
$\ket{0}^{\otimes(n-k)}$, to a state $\ket{\varphi}$ of $k$
qubits. Then an encoding unitary $U_{\rm enc}$ is applied:
\begin{alignat}{5}\label{eq:encode}
\ket{\varphi} & \mapsto \ket{\varphi} \otimes \ket{0}^{\otimes (n-k)} \notag\\
              &\mapsto \ket{\Psi_L}:= U_{\rm enc}\left(\ket{\varphi} \otimes \ket{0}^{\otimes (n-k)}\right).
\end{alignat}
The state $\ket{\Psi_L}$ is called the \textit{encoded state} or the
\textit{logical state}. The unitary $U_{\rm enc}$ acts on the space
formed by the input state $\ket{\varphi}$ and the ancillas
$\ket{0}^{\otimes (n-k)}$. 

In an \emph{entanglement-assisted setup} \cite{Brun2006}, we replace $c$ of
the ancillas by $c$ \emph{pairs} of maximally entangled qubits (\textit{ebits} in short), before applying $U_{\rm enc}$.  Let the state
$\ket{\Psi^+}_{AB}$ be an EPR pair \cite{Einstein1935} shared
between the sender Alice and the receiver Bob.  The encoding operation
performed on Alice's qubits is given by
\begin{alignat}{5}
\ket{\varphi}& \mapsto \ket{\varphi} \otimes \ket{0} ^{\otimes (n-k-c)} \otimes \ket{\Psi^+}_{AB}^{\otimes c} \notag \\
 & \mapsto\left(U_{\rm enc} \otimes I_{B}\right) \left(\ket{\varphi} \otimes \ket{0}^{\otimes (n-k-c)}  \otimes \ket{\Psi^+}_{AB}^{\otimes c}\right).\label{eq:entencode2}
\end{alignat}
Here $I_B$ denotes the identity operator on the receiver's half of the $c$
maximally entangled pairs, \textit{i.\,e.}, the encoded state
consists of $n+c$ qubits in total. The notation
$\ket{\Psi^+}_{AB}^{\otimes c}$ should be understood as reordering
the qubits such that the first $c$ qubits are with the sender Alice.
We assume that the noise only affects the first $n$ qubits sent over the channel, while the half of the shared $c$ ebits that is with the receiver Bob is not affected.

\begin{figure}
\centering
\begin{adjustbox}{width=\hsize}
\begin{quantikz}
\lstick{$\ket{\varphi}$} & [1.8cm]\gate[wires=3]{U} 
\qwbundle
{k} & \gate[3,style={starburst,fill=yellow,draw=red,line
width=1pt,inner xsep=-6pt,inner ysep=-14pt},
label style=black]{\text{noise}} & \gate[wires=4,disable auto height]
{\begin{array}{c} \text{Error} \\ \text{diagnosis} \\ \text{and} \\ \text{recovery} \end{array}} \qw\\
\lstick{$\ket{0}$} & \qwbundle{n-k-c} & \qw & \qw  \\
\makeebit{$\ket{\Phi}_{AB}$} & \qwbundle{c} & \qw & \qw \\
& \qw & \qwbundle{c} & \qw \\
\end{quantikz}
\end{adjustbox}
\caption{The basic structure of an entanglement-assisted quantum error-correcting code. Alice and Bob share the maximally entangled state $\ket{\Phi}_{AB}$ of $c$ maximally entangled qudits. Alice uses her half of the maximally entangled state and $n-k-c$ ancillas in a fixed state $\ket{0}$ in the encoding of her quantum information $\ket{\varphi}$ by the operator $U:=U_{\rm enc}$. The resulting $q^n$-dimensional state passes through the noisy channel. Bob's half of the initial $c$ pairs of maximally entangled states is assumed to be error-free. They will be used in error diagnosis and recovery.}
\label{fig:circuit}
\end{figure}
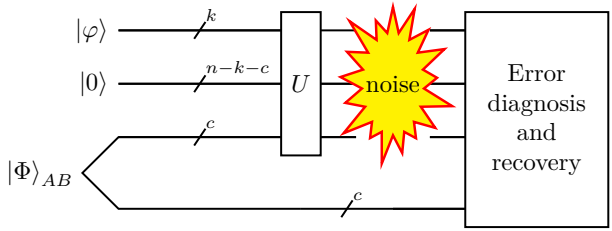

The ebits are prepared \emph{ahead of time} and are assumed to be
error-free, using, \textit{e.\,g.}, entanglement distillation or a
similar procedure. For more details on quantum entanglement, including
its creation and delivery across quantum networks, one can consult
the survey in \cite{Horodecki2009}, the experiment reported in
\cite{Humphreys2018}, or a recent scheme to generate genuine
multipartite entanglement of a large number of qubits in
\cite{Zhang2021}.

In \textit{teleportation}, one perfect ebit is used in tandem with noiseless classical communication to perfectly transmit one qubit. Entanglement-assisted quantum error-correcting codes (EAQECCs) use some noiseless ebits, without classical communication, to transmit quantum information over a noisy quantum channel.

In \textit{superdense coding}, the sender can apply an operation to her half of a pair of maximally entangled states such that, after sending this qubit, the receiver can decode \emph{two} classical bits of information.  An ancilla in a standard QECC can be interpreted as a
\emph{placeholder} for one bit of classical information about any error that has occurred.  Replacing an ancilla with one half of an ebit can, in theory, enable the receiver to extract \emph{two} bits of classical information regarding the errors. This enhances the error-handling capabilities of EAQECCs over the standard counterparts.

One can execute some communication tasks with fewer total resources or better error control by using an EAQECC instead of a combination of a standard QECC and teleportation.  An asymptotic analysis on the
benefits of pre-shared entanglement in quantum communication is available in \cite{Devetak2008}. We view an EAQECC as a finite-length realization. It is in principle possible to approach the
entanglement-assisted quantum capacity by building larger code blocks as shown in \cite{Bowen2002}.

Shared entanglement does not come for free. The cost of
sharing and purifying ebits means that EAQECCs \emph{do not automatically outperform} standard quantum codes in all circumstances. One measure to assess the advantage is the \textit{net rate}, which subtracts the number of ebits required from the number of qubits
transmitted.  In terms of construction via classical error-correcting codes, however, EAQECCs have fewer restrictions, allowing us to use larger families of classical codes.

The notation $\dsb{n,\kappa,\delta; c}_q$ signifies that the quantum
code $\mathcal{Q}$ is a $q$-ary EAQECC that encodes $\kappa$ logical
\textit{qudits} (quantum systems of dimension $q$) into $n$ physical
qudits, with the help of $n-\kappa-c$ ancillas and $c$ pairs of
maximally entangled qudits.  A quantum code with minimum distance $\delta$ can correct up to $\lfloor{(\delta-1)/2\rfloor}$
single-qudit errors. As shown in Figure \ref{fig:circuit}, Alice transmits the $n$ qudits to Bob. He then performs a syndrome measurement on them together with his half of the $c$ pairs of maximally entangled qudits to correct errors and
to retrieve the $\kappa$ logical qudits. The \textit{rate} $\rho$ and the \textit{net rate} $\Bar{\rho}$ of $\mathcal{Q}$ are, respectively,
\begin{equation}\label{eq:rate}
\rho:= \frac{\kappa}{n} \quad\text{and}\quad
\Bar{\rho}:= \frac{\kappa - c}{n}.
\end{equation}
The abbreviated form $\dsb{n,\kappa,\delta}_q$ is used when $c=0$. 

\subsection{Quantum Codes from Classical Codes}\label{subsec:construct}

Let $p$ be a prime and let $s$ be a positive integer. Let $q$ be a prime
power $q=p^{s}$ and let $\mathbb{F}_q$ be the finite field with $q$
elements. The multiplicative group of the nonzero elements of
$\mathbb{F}_q$ is denoted by $\mathbb{F}_q^{*}$. For a positive
integer $m$, we denote by $[m]$ the set $\{1,2,\ldots,m\}$. A code
$\mathcal{C}$ of length $n$ is a nonempty subset of
$\mathbb{F}_q^n$. Its \textit{codewords} are vectors of length $n$
with entries from $\mathbb{F}_q$. The {\it weight} of a vector is the
number of its nonzero entries. Given a nonempty $\mathcal{S} \subseteq \mathbb{F}_q^n$, we denote by ${\rm wt} (\mathcal{S})$ the number 
$\min\{ {\rm wt}(\mathbf{v})\colon \mathbf{v} \in \mathcal{S}, \mathbf{v}
\neq \mathbf{0}\}$.

A code $\mathcal{C}$ is {\it linear} with parameters $[n,k,d]_q$ if it
is a $k$-dimensional subspace of $\mathbb{F}_q^n$ and its {\it minimum distance}, defined to be the smallest of the weights of its nonzero
codewords, is $d$. A $k \times n$ matrix $G$ whose rows form a basis
for $\mathcal{C}$ is a \textit{generator matrix} of $\mathcal{C}$. If
$G = \begin{pmatrix} I_k\; A \end{pmatrix}$, where $I_k$ is the $k
\times k$ identity matrix, we say that $G$ is \textit{in the standard form}.

We equip $\mathbb{F}_q^n$ and $\mathbb{F}_{q^2}^n$ with the Euclidean
and the Hermitian inner products, respectively. Given an arbitrary vector $\mathbf{x}=(x_1,\ldots,x_n)$ and a codeword $\mathbf{c}=(c_1,\ldots,c_n)$ in $\mathcal{C}$, the {\it Euclidean dual} $\mathcal{C}^{\perp}$ of $\mathcal{C}$ is 
\begin{equation}\label{eq:Euclidean}
\mathcal{C}^{\perp} = \left\{\mathbf{x} \in \mathbb{F}_q^n\colon \sum_{i=1}^n x_ic_i=0, 
\text{for all $\mathbf{c} \in \mathcal{C}$} \right\}.
\end{equation}
Analogously, the {\it Hermitian dual} $\mathcal{C}^{\perp_{\rm H}}$ of $\mathcal{C}$ is 
\begin{equation}\label{eq:Hermitian}
\mathcal{C}^{\perp_{\rm H}} = \left\{\mathbf{x} \in \mathbb{F}_{q^2}^n\colon \sum_{i=1}^n x_i c_i^q = 0, 
\text{for all $\mathbf{c} \in \mathcal{C}$} \right\}.
\end{equation}

An $[n,k,d]_q$-code with $k \leq
\lfloor{\frac{n}{2}}\rfloor$ is \textit{self-orthogonal} if it is
contained in its dual. If, moreover, $n=2k$, then the code is
\textit{self-dual}. The notion of the \textit{hull} of a code was
introduced in \cite{Assmus1992} to define the intersection of the code
with its dual. Hence, the \textit{Hermitian hull} of $\mathcal{C}$ is
the code ${\rm Hull}_{\rm H}(\mathcal{C}) = \mathcal{C} \cap
\mathcal{C}^{\perp_{\rm H}}$. A code whose hull is $\{\mathbf{0}\}$
intersects trivially with its dual and is called a \textit{linear complementary dual} (LCD) code. Readers interested to know more
about classical codes may consult \cite{Huffman2021}.

Gottesman formulated the \emph{stabilizer formalism} for QECCs in \cite{Gottesman1997}. It was subsequently expressed in the language of classical coding theory in \cite{Calderbank1998}, triggering fruitful cross-pollination of ideas and results between quantum error control and classical coding theory. A general treatment over any finite field followed in \cite{AK01}. A survey can be found in \cite{Ketkar2006}. The main ingredients are self-orthogonal classical (additive) codes under the (trace) Hermitian inner product. The orthogonality condition imposes constraints on the parameters of the
corresponding quantum codes. Entanglement-assisted schemes enlarge the pool of ingredients to include codes which are not self-orthogonal, but require maximally entangled states as an additional resource.

We recall a general construction route of EAQECCs via the
\emph{non-commuting stabilizers} as explained, with illustrations, in
\cite{BrunHsieh}. For the qubit case, a formal treatment is given in
\cite{Brun2014}. It links arbitrary classical codes over
$\mathbb{F}_4$ as well as pairs of codes over $\mathbb{F}_2$ to qubit
EAQECCs. Extensions to the qudit case, where $q>2$, are given in
\cite{Galindo2019,Galindo2021}. Using $\mathbb{F}_{q^2}$-linear
codes, we have the following construction (see
\cite[Theorem~3]{Galindo2019}).

\begin{prop}[Hermitian construction]\label{prop:two}
Let $\mathcal{C}$ be an $[n,k]_{q^2}$-code, and let
$\mathcal{C}^{\perp_{\rm H}}$ denote its Hermitian dual. Then there
exists an $\dsb{n,\kappa, \delta; c}_q$-code $\mathcal{Q}$ with
\begin{align*}
c &{}= k - \dim_{\mathbb{F}_{q^2}} 
\left(\mathcal{C} \cap \mathcal{C}^{\perp_{\rm H}}\right),\\
\kappa &{}= n-2k+c,\\
\text{and}\quad\delta &{}=\begin{cases}
       {\rm wt}\left(\mathcal{C}^{\perp_{\rm H}}\right), &\text{if
         $\mathcal{C}^{\perp_{\rm H}}\subseteq\mathcal{C}$;}\\
       {\rm wt}\left(\mathcal{C}^{\perp_{\rm H}} \setminus
       \left(\mathcal{C} \cap \mathcal{C}^{\perp_{\rm H}}\right)\right), & \text{otherwise.}\\
\end{cases}
\end{align*}
\end{prop}
We note that the construction includes the case that $\mathcal{C}$ is
contained in its Hermitian dual $\mathcal{C}^{\perp_{\rm H}}$, which
implies $c=0$, \textit{i.\,e.}, the quantum codes do not require
entanglement assistance.  
The case $\mathcal{C}^{\perp_{\rm H}}\subseteq\mathcal{C}$ has not been explicitly
addressed in \cite{Galindo2019}. The resulting codes have $c=2k-n$ and $\kappa=0$. For
codes with $\kappa=0$ and minimum distance $\delta$, by definition the code has to be
\emph{pure}, {\it i\,.e.}, there is no error of weight less than $\delta$ that acts trivially on
the code.

Another construction uses a pair of $\mathbb{F}_q$-linear codes of
equal length, yielding the so-called CSS-like family of EAQECCs
(see \cite[Theorem~4]{Galindo2019}).
\begin{prop}[CSS-like construction]\label{prop:one}
If $\mathcal{C}_i$ is an $[n,k_i,d_i]_q$-code for $i = 1,2$, then there is an $\dsb{n,\kappa,\delta; c}_q$-code $\mathcal{Q}$ with 
\begin{align*}
c &{}= k_1 - \dim\bigr(C_1 \cap C_2^{\perp}\bigr),\\
\kappa &{}= n-(k_1+k_2)+c,\quad\text{and}\\
\delta &{}=\begin{cases}
  \min\bigl\{{\rm wt}\bigl(C_1^{\perp}\bigr),{\rm wt}\bigr(C_2^{\perp}\bigr)\bigr\},&\text{if $C_1^\perp\subseteq C_2$};\\
  \min\bigl\{{\rm wt}\bigl(C_1^{\perp} \setminus (C_2 \cap C_1^{\perp})\bigr),\\
  \text{\phantom{$\min\bigl\{$}}{\rm wt}\bigr(C_2^{\perp} \setminus (C_1 \cap C_2^{\perp})\bigr)
  \bigr\}, & \text{otherwise.}
\end{cases}
\end{align*} 
\end{prop}
Again, when $C_2^\bot \subseteq C_1$, we have $c=0$ and the resulting
code does not require entanglement assistance. The case
$C_1^\perp\subseteq C_2$, resulting in $c=k_1+k_2-n$ and $\kappa=0$, has
not been explicitly addressed in \cite{Galindo2019} either.

The code $\mathcal{Q}$ in Proposition \ref{prop:two} is {\it pure} or {\it nondegenerate}
if $\delta = {\rm wt}(\mathcal{C}^{\perp_{\rm H}}) = d(\mathcal{C}^{\perp_{\rm H}})$. The
code $\mathcal{Q}$ in Proposition \ref{prop:one} is pure whenever $\delta =
\min\{d(C_1^\perp),d(C_2^\perp)\}$.  Otherwise, the code is said to be \emph{impure}, and
it is \emph{pure to distance ${\rm wt}\left(\mathcal{C}^{\perp_{\rm H}}\right)$} or
$\min\{d(C_1^\perp),d(C_2^\perp)\}$, respectively.

Another extremal case of Proposition \ref{prop:two} that has not been
explicitly discussed in the literature arises when one considers the
trivial code $\mathcal{C}=[n,n,1]_{q^2}$. In this case, $\mathcal{C}^{\perp_{\rm H}}$ is the trivial code that contains only the zero codeword. We argue that the distance of the resulting EAQECC with parameters $c=n$ and $\kappa=0$ is $n+1$. As the code uses
$c=n$ maximally entangled states, we are in a situation similar to
superdense coding.  Performing a joint measurement on $n$ qudits
received from the channel and the $n$ qudits from the pre-shared
entanglement, the receiver can distinguish $q^{2n}$ different unitary
operations applied by the channel, corresponding to all errors of
weight at most $n$. For this, we do not require $q$ to be a prime power. In summary, we have the following proposition.
\begin{prop}
For any $q \ge 2$, not necessarily a prime power, there exists an EAQECC $\mathcal{Q} = [\![n,0,n+1;n]\!]_q$.
\end{prop}

\subsection{Our Contributions}

\begin{enumerate}
    \item We establish three propagation rules.
    
    The first rule, given as Theorem \ref{thm:more}, increases $c$, signifying that more entanglement is required. The derived quantum code can send more information without losing anything in terms of error handling. 
    
    Theorem \ref{thm:same} gives the second rule. It keeps $c$ fixed while lengthening the code, reducing its size. If some conditions are met, then the quantum distance may increase. 
    
    The third rule is in Theorem \ref{thm:less}. It decreases $c$ while lengthening the code. There may be a price to pay in terms of smaller distances on some occasions.

    For the last two rules, we have less theoretical control over the quantum distances and, therefore, searches are the next best option.

\item Our propagation rules are applicable to \emph{both} stabilizer QECCs and EAQECCs. Most prior propagation rules were designed for stabilizer QECCs whereas our propagation rules works on nontrivial EAQECCs as well. It is in conducting searches for EAQECCs with excellent parameters that the main advantage of our propagation rules come to the fore. They allow us to control either the distance or the number of ebits. 

\item In the realm of classical coding theory, Theorem \ref{thm:hull} provides a simple proof on the equivalence of $\mathbb{F}_{q^2}$-linear codes with diverse Hermitian hull dimensions for $q>2$. For any $[n,k,d]_{q^2}$-code $\mathcal{C}$ with $\dim({\rm Hull}_{\rm H}(\mathcal{C})) = \ell$, there exists an equivalent $[n,k,d]_{q^2}$-code $\mathcal{C}^{\prime}$ with $\dim({\rm Hull}_{\rm H}(\mathcal{C}^{\prime})) = \ell^{\prime}$ for each $\ell^{\prime} \in \{0,1,\ldots,\ell\}$. This generalizes the result for Hermitian LCD codes in \cite[Section V]{Carlet2018} that considered only the case of $\ell^{\prime}=0$.
\end{enumerate}

Given an $[n,k,d]_{q^2}$-code $\mathcal{C}$, Section \ref{sec:propagation} discusses three linear
algebraic approaches that derive codes whose dimensions of Hermitian hulls vary. In the first two approaches, the derived codes have fixed parameters $[n,k,d]_{q^2}$, while the dimension of the hull decreases. In the third approach, $k$ is fixed, while both $n$ and the dimension of the hull increase by $1$, and the distance $d$ is either fixed or improved by $1$. Section \ref{sec:bound} discusses upper bounds on the parameters of EAQECCs. They are subsequently used collectively in Section \ref{sec:compute} as a measure of goodness to motivate our computational process and results. The parameters of the resulting qubit and qutrit EAQECCs are listed in the tables after the concluding remarks in Section \ref{sec:conclu}.

\section{Three New Propagation Rules}\label{sec:propagation}

This section presents three new propagation rules based on their effects on $c$, which quantifies the amount of entanglement. We start by devising tools from the classical ingredients.

\subsection{Tools from Classical Coding Theory}\label{subsec:tools}

For any vector $\mathbf{v}:=(v_1,\ldots,v_{n}) \in
\mathbb{F}_{q^2}^{n}$, we denote by $\mathbf{v}^q$ the vector
$(v_1^q,\ldots,v_{n}^q)$. Let $\mathcal{C}$ be an $[n,k,d]_{q^2}$-code
with generator matrix $G$, and let $\mathbf{v}_1,\ldots,\mathbf{v}_k$
be the rows of $G$. We use $G^{\dagger}$ to denote the $n \times k$
matrix whose columns are $\mathbf{v}_1^q,\ldots,\mathbf{v}_k^q$. We
call $G^{\dagger}$ the {\it Hermitian transpose} of $G$. As usual,
$\mathbf{x}^{\top}$ and $M^{\top}$ denote the respective transposes of
a vector $\mathbf{x}$ and a matrix $M$.

We recall the relation between a code's generator matrix and its Hermitian hull.
\begin{lem}\label{lem:dimHull}
The dimension of the Hermitian hull is 
\begin{alignat}{5}
\dim({\rm Hull}_{\rm H}(\mathcal{C})) &= k-\rank(GG^{\dagger}).\label{eq:dimHull}
\end{alignat}
\end{lem}
\begin{proof}
  A vector $\mathbf{v}$ is an element of ${\rm Hull}_{\rm H}(\mathcal{C})$ 
  if it is both a codeword of $\mathcal{C}$ and
  $\mathcal{C}^{\perp_{\rm H}}$. The first condition requires
  that $\mathbf{v}$ is in the row span of $G$, that is,
  $\mathbf{v}=\mathbf{u}G$ for some $\mathbf{u}\in\mathbb{F}_{q^2}^k$.
  The second condition requires that $\mathbf{v}$ is in the kernel of
  $G^\dagger$, \textit{i.\,e.}, $\mathbf{v}G^\dagger=\mathbf{0}$. In combination we have
  \begin{alignat}{5}
    {\rm Hull}_{\rm H}(\mathcal{C})=\{\mathbf{v}=\mathbf{u}G\colon
    \mathbf{u}\in\mathbb{F}_q^k \mbox{ and } \mathbf{u}GG^\dagger=\mathbf{0}\}.
  \end{alignat}
  This implies \eqref{eq:dimHull}.
\end{proof}

A {\it monomial matrix} is a square matrix with exactly one nonzero
entry in each row and each column and zeros elsewhere. The matrix is a
{\it permutation matrix} if all of its nonzero entries are $1$.  Based
on these two families of matrices, two equivalence relations among linear
codes can be defined.

\begin{defn}\label{def1}
Let two linear codes $\mathcal{C}_1$ and $\mathcal{C}_2$ with
respective generator matrices $G_1$ and $G_2$ be given. Then the
following statements hold.
\begin{enumerate}
  \item The codes $\mathcal{C}_1$ and $\mathcal{C}_2$ are permutation
    equivalent if there exists a permutation matrix $P$ such that $G_1
    P$ is a generator matrix of $\mathcal{C}_2$. 
  \item The codes $\mathcal{C}_1$ and $\mathcal{C}_2$ are monomially
    equivalent if there exists a monomial matrix $M$ such that $G_1M$
    is a generator matrix of $\mathcal{C}_2$. 
\end{enumerate}
\end{defn}
	
Equivalent codes have the same length, dimension, and minimum
distance. It can be shown (see, \textit{e.\,g.}, \cite[Theorem
  1.6.2]{Huffman2005}) that any linear code is permutation equivalent
to a linear code whose generator matrix is in the standard form. The next
result shows that the respective Hermitian hulls of two permutation
equivalent codes have the same dimension.
	
\begin{lem}\label{lem1H}
Any two permutation equivalent $\mathbb{F}_{q^2}$-linear codes
$\mathcal{C}_1$ and $\mathcal{C}_2$ have 
\[
\dim({\rm Hull}_{\rm H}(\mathcal{C}_1)) = \dim({\rm Hull}_{\rm H}(\mathcal{C}_2)).
\]
\end{lem}

\begin{proof}
Let $\mathcal{C}_1$ and $\mathcal{C}_2$ be permutation equivalent
codes with parameters $[n,k,d]_{q^2}$. Let $G_1$ be a generator matrix of $\mathcal{C}_1$. Hence, there is a permutation matrix $P$ such that $G_2 = G_1 P$ is a generator matrix of $\mathcal{C}_2$. This implies that
\[
G_2G_2^\dagger=G_1P(G_1P)^\dagger = G_1PP^{\top}G_1^\dagger =
G_1G_1^\dagger.
\]
In combination with Lemma \ref{lem:dimHull}, the
conclusion follows.
\end{proof}

Regarding monomially equivalent codes, Carlet {\it et al.} in
\cite{Carlet2018} demonstrated that Hermitian LCD codes over
$\mathbb{F}_q$ exist for all possible parameters when $q>2$. The next
result is a generalization of the results for Hermitian LCD codes in
\cite[Section V]{Carlet2018}.  In our notation, only the case of
$\ell^{\prime}=0$ was considered in the said reference.
\begin{thm}\label{thm:hull}
Let $q>2$ be a prime power and let $\mathcal{C}$ be an
$[n,k,d]_{q^2}$-code with $\dim({\rm Hull}_{\rm H}(\mathcal{C})) =
\ell$. Then there exists an equivalent $[n,k,d]_{q^2}$-code
$\mathcal{C}^{\prime}$ with $\dim({\rm Hull}_{\rm
  H}(\mathcal{C}^{\prime})) = \ell^{\prime}$ for each $\ell^{\prime}
\in \{0,1,\ldots,\ell\}$.
\end{thm}
\begin{proof}
Without loss of generality, we can assume that ${\rm Hull}_{\rm H} (C)$ is an $[n,\ell,d']_{q^2}$-code generated in the standard form by $G_1 =\begin{pmatrix} I_{\ell}\; A \end{pmatrix}$. Moreover, we can choose a generator
matrix for $\mathcal{C}$ of the form
\begin{equation}\label{eq:GenG}
G=\begin{pmatrix}
I_{\ell} & A \\
\mathbf{0} & B
\end{pmatrix}
\end{equation}
and derive
\begin{equation}\label{eq:matprod}
G G^{\dagger} = 
\begin{pmatrix}
I_{\ell} & A \\
\mathbf{0} & B
\end{pmatrix}
\begin{pmatrix}
I_{\ell} & \mathbf{0} \\
A^{\dagger} & B^{\dagger}
\end{pmatrix}
=
\begin{pmatrix}
I_{\ell}+ A A^{\dagger} & A B^{\dagger} \\
B A^{\dagger} & B B^{\dagger}
\end{pmatrix}.
\end{equation}
The submatrices $I_{\ell}+ A A^{\dagger}$ and $A B^{\dagger}$ are zero
since $G_1$ generates the Hermitian hull, which is contained in the
Hermitian dual of $\mathcal{C}$. We have $\rank({B B^{\dagger}}) = k -
\ell$. We now consider the generator matrix
\begin{equation}
G^{\prime} = G \;  \Diag\left(a_1,a_2,\ldots,a_{\ell - \ell^{\prime}},1,\ldots,1\right)
\end{equation}
with $a_j\in\mathbb{F}_{q^2}^*$ and $a_j^{q+1} \neq 1$ for $1 \leq j
\leq \ell - \ell^{\prime}$. These conditions can always be met for $q>2$.
Let
\[
T = \begin{pmatrix}
a_1^{q+1}-1 & 0 & \cdots & 0 \\
0  & a_2^{q+1}-1 & \cdots & 0 \\
\vdots  & \vdots & \ddots & \vdots \\
0 & 0  & \cdots & a_{\ell - \ell^{\prime}}^{q+1}-1\\
\end{pmatrix}.
\]
We verify that $G^{\prime} G^{\prime \dagger}$ is the block-diagonal matrix
\begin{multline}{}
G^{\prime} G^{\prime \dagger}  = 
\begin{pmatrix}
T & \mathbf{0} & \mathbf{0} \\
\mathbf{0} & \mathbf{0} & \mathbf{0} \\
\mathbf{0} & \mathbf{0}  & B B^{\dagger}\\
\end{pmatrix} = \\
\Diag \left(a_1^{q+1}-1,\cdots,a_{\ell - \ell^{\prime}}^{q+1}-1, 0, \cdots,0, 
B B^{\dagger}\right)
\end{multline}
and that $\rank\left(G^{\prime} G^{\prime \dagger}\right) = k -
\ell^{\prime}$. This completes the proof.
\end{proof}

\begin{remark}
Independently and coming from a different motivation, H. Chen derived a similar result to Theorem \ref{thm:hull} in \cite[Corollary 2.2]{Chen2022}. The said corollary was added since version $2$ of his pre-print following a private communication with the first author. 
\end{remark}

We will need the following lemma (see \cite[Theorem 6.32]{Hou}) later.
\begin{lem}\label{lem42}
Let $A$ be an $n\times n$ matrix with rank $s$ over $\mathbb{F}_{q^2}$
such that $A= A^{\dagger}$. Then $A$ is Hermitian congruent to 
\[
\Diag \Big(\underbrace{1,\cdots,1}_s,0,\cdots,0\Big).
\]
\end{lem}
Here two matrices $A$ and $B$ over $\mathbb{F}_{q^2}$ are {\it Hermitian congruent} if there exists a nonsingular matrix $D$ such that $B = D A D^{\dagger}$.

Using Lemma \ref{lem42}, we derive the following result. It enables an
$[n,k]_{q^2}$-code with $\ell$-dimensional Hermitian hull to generate
an $[n+1,k]_{q^2}$-code with $(\ell+1)$-dimensional Hermitian hull. 
\begin{prop}\label{cor:extension}
Let $0\leq \ell<\min\{k,n-k\}$. Given an $[n,k,d]_{q^2}$-code $\mathcal{C}$ with
$\dim({\rm Hull}_{\rm H}(\mathcal{C}))=\ell$ and generator matrix
$G$, one can add one column to $G$ such that the Hermitian hull of the
extended $[n+1,k,d^{\prime}]_{q^2}$-code $\mathcal{C}^{\prime}$ has
dimension $\dim\left({\rm Hull}_{\rm H}(\mathcal{C}^{\prime})\right) = \ell+1$ and minimum distance $d'$, with $d\le d'\le d+1$.
\end{prop}
\begin{proof}
Let $q=p^m$. Let $G$ be a generator matrix of $\mathcal{C}$. Since
$\dim({\rm Hull}_{\rm H}(\mathcal{C}))=\ell$, by Lemma \ref{lem:dimHull}, $\rank\left(GG^{\dagger}\right)=s$, where $s = k-\ell\ge 1$. By Lemma
\ref{lem42}, there exists a nonsingular $k \times k$ matrix $D$ over
$\mathbb{F}_{q^2}$ such that
\begin{equation}\label{eq:thm41H}
DGG^{\dagger}D^{\dagger}= \Diag \Big(\underbrace{1,\cdots,1}_s,0,\cdots,0 \Big). 
\end{equation}
Since $s\ge 1$, the first diagonal entry of the $k\times k$ diagonal
matrix in \eqref{eq:thm41H} must be $1$. Because $D$ is nonsingular,
$DG$ is also a generator matrix of $\mathcal{C}$. Let $\alpha \in
\mathbb{F}_{q^2}$ be such that $\alpha^{q+1} = -1$. Such an $\alpha$
always exists since $\alpha^{q+1}$ runs through $\mathbb{F}_q$ when
$\alpha$ runs through $\mathbb{F}_{q^2}$. Let $G^{\prime}$ be the
$k\times(n+1)$ matrix defined by
\begin{equation}
G^{\prime}=\left(\begin{array}{@{}c|c@{}}
DG & \begin{matrix}
\alpha  \\
0   \\
\vdots  \\
0 \\
\end{matrix}
\end{array}\right).\label{eq:extGenMat}
\end{equation}
Then $G^{\prime}$ generates an $[n+1,k,d^{\prime}]_q$-code
$\mathcal{C}^{\prime}$ with $d\leq d^{\prime}\leq d+1$. One then
verifies that
\begin{alignat}{5}
  G^{\prime}G^{\prime\dagger}
  &=DGG^\dagger G^\dagger+\Diag(\alpha^{q+1},0,\cdots,0)\\
  &= \Diag \Big(0,\underbrace{1,\cdots,1}_{s-1},0,\cdots,0 \Big), 
\end{alignat}
and, hence, $\rank(G^{\prime} \, G^{\prime\dagger}) = s-1$.
The claim about the dimension of the hull follows from Lemma
\ref{lem:dimHull}. As $D$ is invertible, the matrix
\begin{alignat}{5}
    D^{-1}G'=\left(\begin{array}{@{}c|c@{}}
          G & D^{-1}\begin{pmatrix}
               \alpha  \\
               0   \\
               \vdots  \\
                0 \\
          \end{pmatrix}
    \end{array}\right)
\end{alignat}
is a generator matrix for $\mathcal{C}'$ in the desired form.
\end{proof}

We note that the matrix $D$ in \eqref{eq:thm41H} is not unique. Moreover, there are $q+1$ choices for the element $\alpha$, which can be at any of the first $s$ positions. Hence, there are many choices for the additional column in Proposition \ref{cor:extension}.

In \cite{lisonvek2014quantum}, Lison{\v{e}}k and Singh proposed a modified construction of quantum codes by relaxing the self-orthogonality requirement. From a linear code $\mathcal{C}$ that is not Hermitian self-orthogonal, one can obtain a new linear code which is Hermitian self-orthogonal by adding some rows and columns to a generator matrix of $\mathcal{C}$. Inspired by this result, we show that an 
$[n,k]_{q^2}$-code with $\ell$-dimensional Hermitian hull gives rise to an $[n+1,k+1]_{q^2}$-code with $(\ell+1)$-dimensional Hermitian hull.

\begin{prop}\label{pextendrule}
Let $\mathcal{C}$ be an $[n,k,d]_{q^2}$-code with basis $\{\mathbf{a}_1,\cdots,\mathbf{a}_k\}$ and $\dim({\rm Hull}_{\rm H}(\mathcal{C}))=\ell$, with $0\le\ell<\min\{k,n-k\}$. Let $\mathbf{c}$ be a chosen codeword of $\mathcal{C}^{\perp_{\rm H}}\setminus{\rm Hull}_{\rm H}(\mathcal{C})$ such that $\mathbf{c}\mathbf{c}^{\dagger}\neq 0$. Then there exists an $[n+1,k+1,d^{\prime}]_{q^2}$-code $\mathcal{C}^{\prime}$ with
$\dim\left({\rm Hull}_{\rm H}(\mathcal{C}^{\prime})\right) = \ell+1$ and $d^{\prime}=\min\{d,d_0+1\}$, where $d_0$ is the minimum distance of the code generated by $\{\mathbf{a}_1,\cdots,\mathbf{a}_k,\mathbf{c}\}$.
\end{prop}

\begin{proof}
Such a codeword $\mathbf{c}\in\mathcal{C}^{\perp_{\rm H}} \setminus{\rm Hull}_{\rm H}(\mathcal{C})$ with $\mathbf{c}\mathbf{c}^{\dagger}\neq 0$ always exists since $\dim({\rm Hull}_{\rm H}(\mathcal{C}))=\ell$ with $0\le\ell<\min\{k,n-k\}$. Let $G$ be a generator matrix of $\mathcal{C}$ whose rows are $\{\mathbf{a}_1,\cdots,\mathbf{a}_k\}$. Let $\mathbf{c}\mathbf{c}^{\dagger}=\alpha$. Since $\alpha\in\mathbb{F}_q^*$, there exists $\beta\in\mathbb{F}_{q^2}^*$ such that $\beta^{q+1}=-\alpha$. Let $G^{\prime}$ be the
$(k+1)\times(n+1)$ matrix defined by
\begin{equation}\label{eq26}
G^{\prime} := \left(
\begin{matrix}G&\mathbf{0}_{k\times 1}\\ 
\mathbf{c}&\beta \end{matrix}\right).
\end{equation}
Then we obtain the code 
\[
\mathcal{C}^{\prime}=\{(x_1,\cdots,x_{k+1})\cdot G^{\prime} : x_i\in\mathbb{F}_{q^2},i\in[1,\cdots,k+1]\}.
\]
Putting $x_{k+1}$ to be either zero or nonzero, we deduce that the minimum distance of $\mathcal{C}^{\prime}$ is $d^{\prime}=\min\{d,d_0+1\}$. It follows from (\ref{eq26}) that 
\begin{equation}
G^{\prime} G^{\prime\dagger} = 
\begin{pmatrix}
G G^{\dagger} & G \mathbf{c}^{\dagger} \\
\mathbf{c} G^{\dagger} & 0
\end{pmatrix}.
\end{equation}
Since $\mathbf{c} \in \mathcal{C}^{\perp_{\rm H}}$, the column  $G \mathbf{c}^{\dagger}$ is a zero vector. Thus, $\dim\left({\rm Hull}_{\rm H}(\mathcal{C}^{\prime})\right) = k+1-\rank(GG^{\dagger})=\ell+1$.
\end{proof}

\subsection{Propagation Rules}\label{subsec:rules}

The new propagation rules are presented based on how they affect the variable $c$. 

\begin{thm}[More Entanglement]\label{thm:more}
For $q>2$, the existence of a pure $\dsb{n,\kappa,\delta;c}_q$-code
$\mathcal{Q}$, constructed by Proposition \ref{prop:two}, implies the
existence of an $\dsb{n, \kappa+i, \delta; c+i}_q$-code
$\mathcal{Q}^{\prime}$ that is pure to distance $\delta$ for each $i
\in\{1,\ldots,\ell\}$, where $\ell$ is the dimension of the Hermitian
hull of the $\mathbb{F}_{q^2}$-linear code $\mathcal{C}$ that
corresponds to $\mathcal{Q}$.
\end{thm}
\begin{proof}
To confirm the assertion, we apply the Hermitian construction in Proposition \ref{prop:two} on the codes from Theorem \ref{thm:hull} and use Lemma \ref{lem:dimHull} to cover the stated range of parameters.
\end{proof}

\begin{example}
Let $\omega$ be a root of $x^2 + 2x + 2 \in \mathbb{F}_3[x]$ and let
$\mathbb{F}_9= \mathbb{F}_3(\omega)$. The $[29,14,12]_9$-code
$\mathcal{C}$ generated by $G=\begin{pmatrix}I_{14}\;
A\end{pmatrix}$, with $A$ being the matrix 
\begin{small}
\[\let\w\omega
\setlength\arraycolsep{1.8pt}
\begin{pmatrix}
2 &  \w  &  \w^5 & \w^7 & \w^7 & \w^2 & \w & 0 & 0 & \w^5 & \w^6 & \w^3 & \w^3 &  \w & \w^5 \\
\w^5 & \w & 2 & \w^3 & \w & 0 & \w^6 & \w & 0 & \w^6 & 2 & \w^3 & \w^5 & 2 & 2\\
2 & 0 & 0 & \w^2 & 0 & \w^3 & \w & \w^6 & \w & \w^5 & \w^2 & \w^5 & \w^7 & 0 & \w^6 \\
\w^6 & \w^5 & \w^7 & \w & \w^3 & 2 & \w^7 & \w & \w^6 & \w^6 & \w^3 & 1 & \w & 0 & \w^7 \\
\w^7 & \w^3 & \w^3 & \w^5 & \w^3 & \w^2 & 1 & \w^7 & \w & \w^5 & 2 & \w & \w^5 &  \w^7 & 1 \\
1 & 2 & 1 & \w^7 & \w^2 & \w & 1 & 1 & \w^7 & \w^5 & 1 & \w^2 & \w^6 & \w &  \w^6 \\
\w^6 & \w^6 & \w^2 & \w^2 & \w^6 & \w & 1 & 1 & 1 & \w^3 & \w^3 & \w^3 & 1 & \w & \w^6 \\
\w^6 & \w & 1 & \w^3 & \w^3 & \w^3 & 1 & 1 & 1 & \w & \w^6 & \w^2 & \w^2 & \w^6 & \w^6 \\
\w^6 & \w & \w^6 & \w^2 & 1 & \w^5 & \w^7 & 1 & 1 & \w & \w^2 & \w^7 & 1 & 2 & 1 \\
1 & \w^7 & \w^5 & \w & 2 & \w^5 & \w & \w^7 & 1 & \w^2 & \w^3 & \w^5 & \w^3 & \w^3 & \w^7 \\
\w^7 & 0 & \w & 1 & \w^3 & \w^6 & \w^6 & \w & \w^7 & 2 & \w^3 & \w & \w^7 & \w^5 & \w^6 \\
\w^6 & 0 & \w^7 & \w^5 & \w^2 & \w^5 & \w & \w^6 & \w & \w^3 & 0 & \w^2 & 0 & 0 & 2 \\
2 & 2 & \w^5 & \w^3 & 2 & \w^6 & 0 & \w & \w^6 & 0 & \w &  \w^3 & 2 & \w & \w^5\\
\w^5 & \w & \w^3 & \w^3 & \w^6 & \w^5 & 0 & 0 & \w & \w^2 & \w^7 & \w^7 & \w^5 & \w & 2
\end{pmatrix},
\]
\end{small}%
is Hermitian self-orthogonal. The dual has parameters $[29,15,11]_9$. For their respective $(n,k)$ values, both $\mathcal{C}$ and $\mathcal{C}^{\perp_{\rm H}}$ have the best-known minimum distances. We get a $\dsb{29,1,11;0}_3$-code by Proposition \ref{prop:two}. The existence of a $\dsb{29,1+i,11;i}_3$-code for each $1 \leq i \leq 14$ is guaranteed by Theorem \ref{thm:more}.
\end{example}

Theorem \ref{thm:more} allows for the transmission of a larger number of qudits when more pairs of maximally entangled qudits are available, while preserving the minimum distance $\delta$, the total number $n$ of qudits to be sent, as well as the net rate. The main idea is to multiply the columns of the generator matrix by an invertible diagonal matrix to \emph{decrease} the dimension of the Hermitian hull. 

We can use the same approach to try to \emph{increase} the dimension of the Hermitian hull. This yields the following generalization of the Hermitian construction in Proposition \ref{prop:two}.

\begin{thm}\label{thm:minEntanglement}
Let $\mathcal{C}$ be an $[n,k]_{q^2}$-code whose Hermitian dual is $\mathcal{C}^{\perp_{\rm H}}$. Then there
exists an $\dsb{n,\kappa, \delta; c}_q$-code $\mathcal{Q}$ with
\begin{align}
c &{}=\min\left\{\rank\bigl(G\Diag(b_1,\cdots,b_n) G^{\dagger}\bigr)\colon b_i\in\mathbb{F}_q^*\right\},\label{eq:min_entanglement} \\
\kappa &{}= n-2k+c,\notag \\
\text{and}\quad\delta{}&\ge\begin{cases}
   {\rm wt}\left(\mathcal{C}^{\perp_{\rm H}} \setminus \left(\mathcal{C}\cap \mathcal{C}^{\perp_{\rm H}}\right)\right),&\text{if $c>2k-n$;}\\
   {\rm wt}\left(\mathcal{C}^{\perp_{\rm H}}\right),&\text{if $c=2k-n$.}
\end{cases} \label{eq:min_entanglement_distance}
\end{align}
\end{thm}
\begin{proof}
Consider the code $\mathcal{C}'$ generated by a matrix $G' := G \Diag(a_1,\cdots,a_n)$, with $a_i\in\mathbb{F}_{q^2}^*$. By Lemma \ref{lem:dimHull}, 
$\dim({\rm Hull}_{\rm H}(\mathcal{C}')) = k-\rank(G'{G'}^\dagger)$, where
\begin{multline}\label{eq:iso_rank}
\rank(G'{G'}^\dagger)= \\
\rank\left(G \, \Diag(a_1^{q+1},\cdots,a_n^{q+1}) \, G^\dagger\right).
\end{multline}
As $a_i^{q+1}\in\mathbb{F}_q$, it suffices to minimize
\eqref{eq:iso_rank} over all invertible diagonal matrices over
$\mathbb{F}_q$.  By the surjectivity of the norm, given
$b_i\in\mathbb{F}_q$, there exists $b_i\in\mathbb{F}_{q^2}$ with
$b_i^{q+1}=a_i$.

Concerning the minimum distance $\delta$, first we note that multiplying the
coordinates of the code $\mathcal{C}$ with non-zero elements $a_i$
does not change its distance or that of its Hermitian dual
$\mathcal{C}^{\perp_{\rm H}}$.  Moreover, the Hermitian hull
${\rm Hull}_{\rm H}(\mathcal{C}')$ contains the transformed
vectors of ${\rm Hull}_{\rm H}(\mathcal{C})$. Since the Hermitian hull of $\mathcal{C}'$ might be a larger set than that of $\mathcal{C}$, we have
\begin{align*}
   {\rm wt}\left(\mathcal{C'}^{\perp_{\rm H}} \setminus \bigl(\mathcal{C'}\cap \mathcal{C'}^{\perp_{\rm H}}\bigr)\right) \ge
   {\rm wt}\left(\mathcal{C}^{\perp_{\rm H}} \setminus \bigl(\mathcal{C}\cap \mathcal{C}^{\perp_{\rm H}}\bigr)\right).
\end{align*}
The second part of
\eqref{eq:min_entanglement_distance} applies in the extremal case when $\mathcal{C}'^{\perp_{\rm H}}\subseteq\mathcal{C}'$,
$c=2k-n$, and $\kappa=0$.
\end{proof}
We do not have an efficient method to determine an equivalent code
$\mathcal{C}^{\prime}$ that minimizes \eqref{eq:min_entanglement}.

In the extremal case of $c=0$, we obtain a Hermitian self-orthogonal
code by finding particular solutions to a linear system.
\begin{thm}
A given $[n,k,d]_{q^2}$-code $\mathcal{C}$ is equivalent to a Hermitian self-orthogonal code if and only if there is a vector $\mathbf{b}\in\mathbb{F}_q^n$, with $b_i \ne 0$, such that
\begin{equation}\label{eq:punctureCode}
\sum_{i=1}^n b_i x_i y_i^q = 0 \mbox{ for all } \mathbf{x},\mathbf{y}\in\mathcal{C}.
\end{equation}
\end{thm}
\begin{proof}
Let $\mathcal{C}$ be equivalent to a Hermitian self-orthogonal code
$\mathcal{C}^{\prime}$. Without loss of generality, we can assume that
$\mathcal{C}^{\prime}$ is obtained from $\mathcal{C}$ via
multiplication of the coordinates by $a_i
\in\mathbb{F}_{q^2}^*$. Setting $b_i=a_i^{q+1}\in\mathbb{F}_q^*$
implies \eqref{eq:punctureCode}.  On the other hand, if
\eqref{eq:punctureCode} has a solution, then, by the surjectivity of
the norm, we find $a_i\in\mathbb{F}_{q^2}^*$ with
$b_i=a_i^{q+1}$. Multiplying the coordinates of $\mathcal{C}$ by $b_i$
yields an equivalent Hermitian self-orthogonal code, provided that
$a_i$ and $b_i$ are nonzero for all $i$.
\end{proof}
The linear space of all solutions $\mathbf{b}\in\mathbb{F}_q^n$ of
\eqref{eq:punctureCode} (\textit{i.e.}, allowing zero-coordinates as well) is known as a \emph{punctured code} \cite{Rains99}.

Swapping the roles of $\mathcal{C}$ and $\mathcal{C}^{\perp_{\rm H}}$ in Proposition \ref{prop:two} enables us to better control the quantum distance. We use this approach in the proof of the next result.

\begin{thm}[Same Entanglement]\label{thm:same}
If there exists a pure $\dsb{n,\kappa,\delta;c}_q$-code $\mathcal{Q}$ with $\kappa>0$, $c>0$
obtained by the Hermitian construction in Proposition \ref{prop:two}, then there exists an
$\dsb{n+1, \kappa-1, \delta'; c}_q$-code $\mathcal{Q}^{\prime}$ that is pure to distance
$\delta'$ with $\delta\le\delta'\le\delta+1$.
\end{thm}
\begin{proof}
Let $\mathcal{C}^{\perp_{\rm H}}$ be the linear code with parameters $[n,n-k,\delta]_{q^2}$ used in the Hermitian construction of $\mathcal{Q}$. The dimension of the Hermitian hull is 
\[
\dim(\mathcal{C}\cap\mathcal{C}^{\perp_{\rm H}})=\ell=n-k-c.
\]
Applying Proposition \ref{cor:extension} to $\mathcal{C}^{\perp_{\rm H}}$ yields an $[n+1,n-k,\delta']_{q^2}$-code ${\mathcal{C}^{\prime}}^{\perp_{\rm H}}$. The dimension of its Hermitian hull is $\ell+1$. Applying the Hermitian construction to $\mathcal{C}^{\prime}$ gives us an $\dsb{n+1,\kappa-1,\delta';c}_q$-code $\mathcal{Q}^{\prime}$.
\end{proof}

Applying the Hermitian construction on the $[n+1,k,d']_{q^2}$-code $\mathcal{C}^{\prime}$ from Proposition \ref{cor:extension} produces an $\dsb{n+1,\kappa-1,\delta';c}_q$-code $\mathcal{Q}^{\prime}$. On the original classical code
$\mathcal{C}$, the outcome is an $\dsb{n,\kappa,\delta;c}_q$-code $\mathcal{Q}$. The minimum distance $\delta'$ of $\mathcal{Q}^{\prime}$ depends on how the extended code $\mathcal{C}^{\prime}$ is built, as illustrated in the following example.
\begin{example}\label{ex:thm10}
Let $\omega$ be a root of $x^2 + 2x +2 \in \mathbb{F}_3[x]$ and let
$\mathbb{F}_9=\mathbb{F}_3(\omega)$. Let $\mathcal{C}$ be the
$[5,4,2]_9$-code generated by $G=\begin{pmatrix} I_4\;B \end{pmatrix}$
with $B=\begin{pmatrix} 2\; 2\;2\;2 \end{pmatrix}^{\top}$. Extending
the matrix $G$ by the column
$\begin{pmatrix}1\; w^3\; w^2\; w^7\end{pmatrix}^{\top}$, we obtain a $[6,4,3]_9$-code $\mathcal{C}^{\prime}$ with 
$\dim{\rm Hull} (\mathcal{C}^{\prime})=1$. Since $\mathcal{C}^{\prime}$ is an MDS code, its Hermitian dual
${\mathcal{C}^{\prime}}^{\perp_{\rm H}}$ has parameters $[6,2,5]_9$.
Using the code $\mathcal{C}^{\prime}$ in Proposition \ref{prop:two} results in a pure $\dsb{6,1,5;3}_3$-code which is optimal by the bounds in the next section. Its net rate is $-1/3$ and it improves on the distance of the distance-optimal quantum code $\dsb{6,1,3}_3$ by $2$.
\end{example}

Based on the derived $[n+1,k+1,d']_{q^2}$-code in Proposition \ref{pextendrule}, we have the following result which allows for the transmission of the same amount of quantum information using a smaller number of pairs of maximally entangled qudits.

\begin{thm}[Less Entanglement]\label{thm:less} The existence of a pure
  $\dsb{n,\kappa,\delta;c}_q$-code $\mathcal{Q}$, constructed from an
  $\mathbb{F}_{q^2}$-linear code $\mathcal{C}$ based on Proposition \ref{prop:two},
  implies the existence of an $\dsb{n+1, \kappa,\delta^{\prime}; c-1}_q$-code
  $\mathcal{Q}^{\prime}$ with $\delta^{\prime} \leq \delta$. 
\end{thm}

\begin{proof}
The desired result follows from Proposition \ref{pextendrule} by a method analogous to the one in the proof of Theorem \ref{thm:same}.
\end{proof}

In Theorem \ref{thm:less}, the pure minimum distance of the resulting EAQECC is determined
by the choice of the codeword $\mathbf{c}$, which was defined earlier in Proposition
\ref{pextendrule}. To construct a new EAQECC with good minimum distance, one can try all
such codewords and then select an EAQECC with the largest minimum distance from all
resulting EAQECCs. The next example illustrates such an implementation.

\begin{example}
Let $\omega$ be a root of $x^2 + 2x +2 \in \mathbb{F}_3[x]$ and let $\mathbb{F}_9=\mathbb{F}_3(\omega)$. Let $\mathcal{C}$ be the $[16, 5, 8]_9$-code with generator matrix $G$ given by
\begin{small}
\[\let\w\omega
\setlength\arraycolsep{2.2pt}
\begin{pmatrix}
1 & 0 & 0 & 0 & \w^6 & 1 & \w^6 & \w^5 & 0 & \w^6 & \w^6 & 1 & \w^7 & 1 & \w & 1 \\
0 & 1 & 0 & 0 & \w^7 & 2 & \w & \w^2 & 0 & \w & 0 & 0 & 1 & \w^5 & \w^2 & \w^3 \\
0 & 0 & 1 & 0 & 2 & 1 & 1 & 1 & 0 & \w^5 & \w^3 & \w & 2 & \w^2 & \w^2 & \w^3 \\
0 & 0 & 0 & 1 & \w^2 & 1 & \w^7 & \w^2 & 0 & 2 & \w & \w^2 & \w & \w^6 & \w & 1 \\
0 & 0 & 0 & 0 & 0 & 0 & 0 & 0 & 1 & 2 & 1 & 2 & 1 & 2 & 1 & 2
\end{pmatrix}.
\]
\end{small}
The Hermitian dual code $\mathcal{C}^{\perp_{\rm H}}$ has parameters $[16,11,5]_9$ and the Hermitian hull ${\rm Hull}_{\rm H}(\mathcal{C})$ is a $[16,3,12]_9$-code. Following the proof of Proposition \ref{pextendrule}, we select a codeword
\begin{equation}\label{eq:wordc}
\let\w\omega
\setlength\arraycolsep{1.6pt}
\mathbf{c}:=
\begin{pmatrix}
1 & \w^7 & \w & \w^5 & \w^6 & \w^2 & 0 & \w^5 & \w^3 & 2 & \w^2 & 1 & \w^2 & \w^5 & \w^3 & \w^2
\end{pmatrix}
\end{equation}
of weight $15$ in $\mathcal{C}^{\perp_{\rm H}} \setminus 
{\rm Hull}_{\rm H} (\mathcal{C})$ such that $\mathbf{c}\mathbf{c}^{\dagger} = 1$ to obtain the $[17,6,8]$-code $\mathcal{C}^{\prime}$ whose generator matrix is
\begin{equation*}
G^{\prime}=\left(
\begin{matrix}G & \mathbf{0}_{5\times 1}\\ 
\mathbf{c}& \omega \end{matrix}\right).
\end{equation*}
The Hermitian dual $\mathcal{C}^{\prime \perp_{\rm H}}$ has parameters $[17,11,5]_9$ and ${\rm Hull}_{\rm H} (\mathcal{C}^{\prime})$ is a $[17,4,10]_9$-code.

We now switch perspective and use the $[16,5,8]_9$-code as the $\mathcal{C}^{\perp_{\rm H}}$, instead of as the $\mathcal{C}$, in Proposition \ref{prop:two} to construct a pure $\dsb{16,2,8;8}_3$-code $\mathcal{Q}$. Using the derived $[17,6,8]$-code $\mathcal{C}^{\prime}$ as the $\mathcal{C}^{\perp_{\rm H}}$ in Proposition \ref{prop:two} leads to a pure $\dsb{17,2,8;7}_3$-code $\mathcal{Q}^{\prime}$.

If we have chosen as our codeword $\mathbf{c}$ the vector
\[
\let\w\omega
\setlength\arraycolsep{1.6pt}
\mathbf{c}:=
\begin{pmatrix}
\w^7 & \w^5 & \w & 0 & \w^5 & \w^5 & 0 & 2 & \w^7 & \w^6 & \w^2 & 2 & \w & 2 & \w^2 & \w^3
\end{pmatrix}
\]
of weight $14$, instead of the one in (\ref{eq:wordc}), then the resulting $\mathcal{C}^{\prime}$ would have parameters $[17,6,7]_9$. The constructed pure quantum codes $\mathcal{Q}$ and $\mathcal{Q}^{\prime}$ would have parameters $\dsb{16,2,7;8}_3$ and $\dsb{17,2,7;7}_3$, respectively. This highlights the importance of choosing $\mathbf{c}$ such that $d^{\prime}=d$, that is, $d_0 \geq d-1$, in Proposition \ref{pextendrule}.
\end{example}

The idea of Lison{\v{e}}k and Singh in \cite{lisonvek2014quantum} is to start with a length $n$ classical code with a large hull. One then carefully selects a codeword so that it can be used to extend the length by $1$ and prevent the quantum distance from deteriorating. Our idea here is similar. The advantage is that we have more freedom in choosing the codeword that may lead to a better quantum distance. The two approaches coincide when the classical ingredient $\mathcal{C}$ is $k$-dimensional over $\mathbb{F}_{q^2}$ and its hull has dimension $k-1$.

\section{Upper Bounds}\label{sec:bound}

There is a vast literature on EAQECCs constructed via classical maximum distance separable (MDS) code and (Hermitian) LCD codes. Their parameters and excellent properties allow for a straightforward derivation of the parameters of the corresponding quantum codes. The length of MDS
codes, however, are constrained by the cardinality of the underlying finite fields. Using classical MDS codes over $\mathbb{F}_4$ for the qubit case and $\mathbb{F}_9$ for the qutrit case offer limited insights beyond very small lengths.

The Singleton bound for an $\dsb{n, \kappa, \delta; c}_q$-code $\mathcal{Q}$ in
\cite[Corollary 9]{GraHubWin2022} reads 
\begin{align}
  \kappa &\le c+\max\{0,n - 2 \delta + 2\},\label{eq:QMDS_small_distance}\\
  \kappa &\le n-\delta+1,\label{eq:QMDS_trivial}\\
  \kappa
  &\le\frac{(n-\delta+1)(c+2\delta-2-n)}{3\delta-3-n} \mbox{, if } \delta-1 \ge \frac{n}{2}. \label{eq:QMDS_large_distance}
\end{align}
Codes attaining the bound \eqref{eq:QMDS_small_distance} for
$\delta-1 \le\frac{n}{2}$ or the bound \eqref{eq:QMDS_large_distance} for
$\delta-1 \ge\frac{n}{2}$ with equality are called MDS EAQECCs. We note that without the bound \eqref{eq:QMDS_trivial}, the upper bound on the dimension
$\kappa$ would be linear in the \textit{a priori} unbounded number $c$
of maximally entangled pairs of qudits.

To our knowledge, most known families of MDS EAQECCs, \textit{e.\,g.}, those presented in \cite{Guenda2017,Chen2021, Chen2021a,Fang2020,Luo2019}, were built by applying Propositions \ref{prop:two} and \ref{prop:one} on suitably chosen classical MDS codes. In \cite{Guenda2017}, an $[n,k,d]_{q^2}$-code
$\mathcal{C}$, whose Hermitian dual is an $[n,n-k,d^{\prime}]_{q^2}$-code, is used in Proposition \ref{prop:two} to yield two EAQECCs with parameters
\begin{alignat}{5}
&\dsb{n, k-\dim({\rm Hull}_{\rm H} (\mathcal{C})),d; 
n-k-\dim({\rm Hull}_{\rm H}(\mathcal{C}))}_q, \\
&\dsb{n,n-k-\dim({\rm Hull}_{\rm H}(\mathcal{C})),d^{\prime}; 
k-\dim({\rm Hull}_{\rm H}(\mathcal{C}))}_q.
\end{alignat}
From an $[n,k,n-k+1]_{q^2}$-MDS code $\mathcal{C}$ and its $[n,n-k,k+1]_{q^2}$-Hermitian
dual $\mathcal{C}^{\perp_{\rm H}}$ with
$\dim({\rm  Hull}_{\rm H}(\mathcal{C})) = 
\dim({\rm  Hull}_{\rm H}(\mathcal{C}^{\perp_{\rm H}}))=\ell$, one obtains EAQECCs with parameters
\begin{alignat}{5}
 &\dsb{n,k-\ell,n-k+1;n-k-\ell}_q \mbox{ and }\\
 &\dsb{n,n-k-\ell,k+1;k-\ell}_q.
\end{alignat}
In general, only the code with distance $d\le\frac{n}{2}$ is an
MDS EAQECC, whereas, for $d>\frac{n}{2}$ and $\kappa\le c <n-k$, the bound \eqref{eq:QMDS_large_distance} cannot be achieved with equality.

As shown by Grassl, Huber, and Winter in \cite[Theorem 7]{GraHubWin2022}, any pure $\dsb{n,\kappa,\delta;c}_q$-code obeys the bounds
\begin{equation}\label{eq:old_QMDS}
2 \delta \le n+c-\kappa+2.
\end{equation}
We show that this bound also applies to EAQECCs that can be obtained by Propositions \ref{prop:two} and \ref{prop:one}.

\begin{thm}\label{thm:bound55}
For any $\dsb{n, \kappa, \delta; c}_q$-code $\mathcal{Q}$ obtained by the Hermitian construction in Proposition~\ref{prop:two}, we have
\begin{alignat}{5}
  2 \delta \leq n + c - \kappa +2.
\end{alignat}
\end{thm}
\begin{proof}
Corresponding to the $\dsb{n, \kappa, \delta; c}_q$-code $\mathcal{Q}$, there exists an 
$\left[n, n-\kappa-\ell \right]_{q^2}$-code $\mathcal{C}$ such that $\dim_{\mathbb{F}_{q^2}}\left({\rm Hull}_{\rm H}\left(\mathcal{C}\right)\right) =\ell$. If ${\rm Hull}_{\rm H}(\mathcal{C})$ has generator matrix $\begin{pmatrix} I_\ell & R
\end{pmatrix}$, then $\mathcal{C}^{\perp_{\rm H}}$ has generator matrix 
\[
\begin{pmatrix}
	I_{\ell} & R\\
	\mathbf{O}_{\kappa \times \ell} & A
\end{pmatrix}.
\]
The code generated by 
\[
\begin{pmatrix}
	\mathbf{O}_{ \kappa  \times \ell} & A
\end{pmatrix}
\]
is a subset of $\mathcal{C}^{\perp_{\rm H}} \setminus {\rm Hull}_{\rm H}(\mathcal{C})$ and
has parameters $[n, \kappa, \geq \delta]_{q^2}$. Hence, the linear code
generated by the matrix $A$ has parameters
\[
\left[n- \ell, \kappa , \geq \delta \right]_{q^2}.	
\]
By the classical Singleton bound, we arrive at
\[
\delta \leq n - \ell -\kappa+ 1.
\]
Since $c=n-\kappa- 2\ell$, we have
\[
\delta \leq n - \frac{n-\kappa-c}{2} -\kappa+ 1 \iff 2 \delta \leq n+c-\kappa+2.
\]
\end{proof}

The codes in the CSS-like subfamily obeys the bound
\eqref{eq:old_QMDS} as well.
\begin{thm}\label{thms}
For $i=1,2$, let $\mathcal{C}_i$ be an $[n,k_i]_q$-code. Let
$\kappa=n-(k_1+k_2)+c$. For any $\dsb{n ,\kappa , \delta; c}_q$-code $\mathcal{Q}$ obtained by the CSS-like construction in Proposition~\ref{prop:one}, we have
\begin{alignat}{5}
 2 \delta \leq n + c - \kappa + 2.
\end{alignat}
\end{thm}
	
\begin{proof}
By Proposition \ref{prop:one}, there exist two linear codes
$\mathcal{C}_1$ and $\mathcal{C}_2$ with respective parameters $[n,k_1]_q$ and
$[n,k_2]_q$, where $k_2 = n- \kappa + c- k_1$. We denote by $\Delta$ the code $\mathcal{C}_1 \cap \mathcal{C}_2^{\perp}$ and
let $\ell=\dim_{\mathbb{F}_q}(\Delta)$. Let $\Delta$ be generated by $\begin{pmatrix} 
I_\ell & R \end{pmatrix}$. Let 
\[
\begin{pmatrix}
I_{\ell} & R\\
\mathbf{O}_{(n-\kappa_2-\ell) \times \ell} & A
\end{pmatrix}
\]
generate $\mathcal{C}_2^{\perp}$. The $[n, n-k_2-\ell, \geq \delta]_q$-code generated by
$\begin{pmatrix}
\mathbf{O}_{(n-k_2-\ell) \times \ell} & A
\end{pmatrix}$ is a subset of $\mathcal{C}_2^{\perp} \setminus
\Delta$. Hence, there exists an $[n- \ell, n-k_2-\ell, \geq
  \delta]_q$-code generated by the matrix $A$. By the
Singleton bound, we infer that $\delta \leq k_2+1$.  In a similar manner, starting from the $\beta$-dimensional code $\Gamma = \mathcal{C}_2 \cap \mathcal{C}_1^{\perp}$, one derives an $[n-\beta, n-k_1-\beta,
  \geq \delta]_q$-code, with $\delta \leq k_1+1$.  Combining the two
inequalities gives us $2 \delta \leq k_1 + k_2 + 2 = n+ c- \kappa + 2$, as promised.
\end{proof}

All these Singleton-type bounds are independent of the alphabet size
$q$.  The classical bound of Griesmer from \cite{Griesmer1960} leads
to a sharper upper bound for lengths $n > q^2+1$.

\begin{thm}
For any $\dsb{n, \kappa, \delta; c}_q$-code $\mathcal{Q}$ obtained by the CSS-like construction in Proposition~\ref{prop:one}, we have 
\[
\frac{n+ \kappa + c}{2} \geq \sum_{i=0}^{\kappa-1} \left\lceil\frac{\delta}{q^i}\right\rceil.
\]
\end{thm}
\begin{proof}
Let $k_2 = n-\kappa+ c- k_1$. By the proof of Theorem \ref{thms},
there exist two linear codes $\mathcal{A}$ and $\mathcal{B}$ with
respective parameters $[n- \ell, n-k_2-\ell, \geq \delta]_q$ and 
$ [n-\beta,n-k_1-\beta, \geq \delta]_q$, where $\ell$ and $\beta$ are the $\mathbb{F}_q$-dimensions of $\mathcal{C}_1 \cap \mathcal{C}_2^{\perp}$ and $\mathcal{C}_2 \cap \mathcal{C}_1^{\perp}$. Let $G_1$ and $G_2$ be
generator matrices of $\mathcal{C}_1$ and $\mathcal{C}_2$,
respectively. Then the dimension of the solution space of $G_1
G_2^{\top} \mathbf{x}^{\top} = \mathbf{0}$ is $k_2-
\rank(G_1G_2^{\top})$. Since
\[
G_1 G_2^{\top} \mathbf{x}^{\top} = G_1 (\mathbf{x}G_2)^{\top} = \mathbf{0},
\]
we have $k_2-\rank(G_1G_2^{\top})=\beta$. Employing the method
analogous to the one we have just used, we arrive at  
\[
k_1-\rank(G_2G_1^{\top})= \ell \implies k_2-\beta=k_1-\ell.
\]
We note that 
\begin{equation}\label{eq:EA_10}
	c=\rank(G_1G_2^{\top})=k_1-\ell=k_2-\beta.
\end{equation}
Applying the Griesmer bound to $\mathcal{A}$ and $\mathcal{B}$ gives
us
\begin{align}
n- \ell & \geq \sum_{i=0}^{n-k_2-\ell-1} \left\lceil\frac{\delta}{q^i}\right\rceil = \sum_{i=0}^{n-(k_1+k_2)+c-1} \left\lceil\frac{\delta}{q^i}\right\rceil \mbox{ and}  \label{eq:EA_11}\\
n- \beta & \geq \sum_{i=0}^{n-k_1-\beta-1} \left\lceil\frac{\delta}{q^i}\right\rceil = \sum_{i=0}^{n-(k_1+k_2)+c-1} \left\lceil\frac{\delta}{q^i}\right\rceil. \label{eq:EA_12}
\end{align}
Since $\ell=k_1-c$ and $\beta=k_2-c$, it follows from (\ref{eq:EA_11})
and (\ref{eq:EA_12}) that
\begin{multline*}
2n-\ell-\beta = 
2n-(k_1+k_2)+2c=  \\
\geq 2 \sum_{i=0}^{n-(k_1+k_2)+c-1} \left\lceil\frac{\delta}{q^i} \right\rceil.
\end{multline*}
The conclusion follows from $k_2=n-\kappa+c-k_1$.
\end{proof}

\begin{thm}{\rm \cite{Li2015}}\label{thm:Griesmer}
For any $\dsb{n,\kappa, \delta; c}_q$-EAQECC obtained by the Hermitian
construction of Proposition~\ref{prop:two}, we have
\begin{equation}\label{eq:GriesmerH}
\frac{n+\kappa+c}{2} \geq \sum_{i=0}^{\kappa-1} \left\lceil\frac{\delta}{q^{2i}}\right\rceil.
\end{equation}
\end{thm}
\begin{proof}
We have an 
$[n-\ell,\kappa, \geq \delta]_{q^2}$-code with $\ell=\frac{n-\kappa-c}{2}$ from the proof of Theorem
\ref{thm:bound55}. By the Griesmer bound,
\[
n-\ell =\frac{n+\kappa+c}{2} \geq \sum_{i=0}^{k-1}\left\lceil\frac{\delta}{q^{2i}}\right\rceil.
\]
\end{proof}

\section{Computational Results}\label{sec:compute}

The results we have derived as well as previously available tools can now be used to search for good entanglement-assisted (EA) qubits and qutrit. 

The simplest approach would have been to apply Proposition \ref{prop:two} on $\mathbb{F}_4$ and $\mathbb{F}_9$-linear codes in the current {\tt MAGMA BKLC} database of codes with best-known minimum distances \cite{Bosma1997,Grassl:codetables}. Most codes in the database are LCD codes or codes with small Hermitian hulls. For qutrit codes, in light of Theorems \ref{thm:hull} and \ref{thm:more}, we prefer codes with large Hermitian hulls, \textit{e.\,g.}, specially crafted quasi-cyclic codes with large Hermitian hulls based on the construction method in \cite[Section III]{ELOS19}. Such classical codes yield EAQECCs with a wider range of parameters. The parameters also depend on the minimum distances of their respective dual codes. The choice of classical codes to record in the said database does not take the above into consideration. One can switch the role of the code and its dual in Proposition \ref{prop:two} so that the code from the database provide information on the minimum distance.

Theorem \ref{thm:same} yields good codes on numerous occasions. Determining the matrix $D$, however, is time consuming and the resulting parameters are often already covered by the other construction approaches. Computational evidences indicate that the benefit from applying Theorem \ref{thm:same} occurs when $d' = d+1$. Replacing the diagonal matrix on the right hand side of (\ref{eq:thm41H}) by a matrix of rank $s$ sometimes allows for a more efficient randomized procedure to find a suitable matrix $D$ that eventually leads to a good qutrit code.

For qutrit codes, we use Theorem
\ref{thm:minEntanglement} to determine the minimum number of maximally entangled pairs $c_{\text{min}}$, \textit{e.\,g.}, by exhaustive search, or use a randomized search to find a smaller value.

We provide the parameters of the best-performing qubit, for lengths $3
\leq n \leq 64$, and qutrit, for lengths $3 \leq n \leq 36$, of
EAQECCs that we can explicitly construct in Tables \ref{table:qubit}
and \ref{table:qutrit}. Among the parameters
$\dsb{n,\kappa,\delta;c}_q$, all other parameters being equal, we
record the smallest $n$, the largest $\kappa$, the largest $\delta$,
and the smallest $c$ for $q \in \{2,3\}$. The tables are compressed
using the propagation rules in this paper and those given in
\cite{Galindo2019,Galindo2021,GraHubWin2022}. For ease of reference we
list the following eight propagation rules. The first four rules are
trivial. Rule $(5)$ is obtained by erasing one position of the
original code.  Rules $(6)$ to $(8)$ come from Theorem \ref{thm:more},
\cite[Theorem 7]{Luo2022}, and \cite[Theorem 8]{Luo2022},
respectively.
\begin{enumerate}
\item[(1)] length extension: $[\![n,\kappa,\delta;c]\!]_q \longrightarrow
      [\![n+1,\kappa,\delta;c]\!]_q$.
\item[(2)] subcode: $[\![n,\kappa,\delta;c]\!]_q \longrightarrow
      [\![n,\kappa-1,\delta;c]\!]_q$.
\item[(3)] smaller distance: $[\![n,\kappa,\delta;c]\!]_q \longrightarrow
      [\![n,\kappa,\delta-1;c]\!]_q$. 
\item[(4)] requiring more entanglement:\\
$[\![n,\kappa,\delta;c]\!]_q \longrightarrow
      [\![n,\kappa,\delta;c+1]\!]_q$. 
\item[(5)] puncturing, assuming $\delta>1$ and $c<n-\kappa$:\\ $[\![n,\kappa,\delta;c]\!]_q\longrightarrow
      [\![n-1,\kappa,\delta-1;c]\!]_q$.
\item[(6)] increasing the dimension of a pure $q$-ary quantum code
  with $q>2$ by using extra entanglement, provided that $c \leq n-\kappa-2$:\\ $[\![n,\kappa,\delta;c]\!]_q\longrightarrow
      [\![n,\kappa+1,\delta;c+1]\!]_q$.
\item[(7)] reducing the length by using extra entanglement, provided that $c \leq n-\kappa-2$:\\ $[\![n,\kappa,\delta;c]\!]_q \longrightarrow
      [\![n-1,\kappa,\delta;c+1]\!]_q$.
\item[(8)] shortening pure quantum code:\\
$[\![n,\kappa,\delta;c]\!]_q \longrightarrow [\![n-1,\kappa+1,\delta-1;c]\!]_q$.
\end{enumerate}

The parameters that we have determined in this work can be found in
the online record of the bounds on the minimum distance of
entanglement-assisted quantum codes \cite{Grassl:EAQECCtables}.

\section{Concluding Remarks}\label{sec:conclu}

The use of pre-shared entanglement in quantum error control raises
questions. How do entanglement-assisted QECCs compare to other QECCs that draw on different resources? In what setups can they be more useful than the others? The enhanced rate or better error-handling capability offered by EAQECCs must be paid for by the additional cost of pre-shared entanglement. On the more practical front, one asks how to best share ebits and how many of them to share.

It is possible for the net rate $\Bar{\rho}(Q)$ to be zero or
negative. 
Can such a code be useful in practice? Since shared entanglement $\ket{\phi}_{AB}$ is independent of the message $\ket{\varphi}$, it can be prepared ahead of time and stored to be used as and when needed. In a quantum network, where usage varies over time, Alice and Bob can use periods of low usage to accumulate ebits. These can then be utilized to increase the transmission rate without trading off on the error-correcting power when the network usage grows higher.

Codes with positive net rate can be used as building blocks in the
construction of \textit{catalytic quantum codes}, leading to the
quantum analogue of highly-efficient classical codes such as Turbo and LDPC codes \cite{Wilde2014}.

The quantum setup provides a rich ground for coding theorists of the more classical mould to venture into topics hitherto less explored. Instead of focusing on quantum codes that meet the analogue of the Singleton bound, for example, constructing qubit and qutrit codes that have better chances of being implemented in actual quantum devices and networks could take a more focal position.

We identify the following open directions for further investigation. 
\begin{enumerate}
    \item Establish sharper lower and upper bounds on the parameters of best EAQECCs, especially for qubit and qutrit codes.
    \item Find the quantum code with the largest rate for a specified quantum distance and hull dimension. The duality can be chosen among suitable choices of inner products, depending on the construction routes.
    \item In the classical setting, given a length $n$ and dimension $k$, construct a code with the largest hull and optimal dual distance.
\end{enumerate}

\begin{acknowledgments}
The authors would like to thank Hao Chen and Tania Sidana for comments
on earlier versions of this manuscript.

G. Luo, M. F. Ezerman, and S. Ling are supported by Nanyang
Technological University Grant 04INS000047C230GRT01.

M. Grassl acknowledges support by the Foundation for Polish Science
(IRAP project, ICTQT, contract no. 2018/MAB/5, co-financed by EU
within Smart Growth Operational Programme).
\end{acknowledgments}

\bibliographystyle{plainnat}
\bibliography{EAQECC}

\begin{thebibliography}{37}
\providecommand{\natexlab}[1]{#1}
\providecommand{\url}[1]{\texttt{#1}}
\expandafter\ifx\csname urlstyle\endcsname\relax
  \providecommand{\doi}[1]{doi: #1}\else
  \providecommand{\doi}{doi: \begingroup \urlstyle{rm}\Url}\fi

\bibitem[Ashikhmin and Knill(2001)]{AK01}
Alexi Ashikhmin and Emanuel Knill.
\newblock Nonbinary quantum stabilizer codes.
\newblock \emph{IEEE Transactions on Information Theory}, 47\penalty0
  (7):\penalty0 3065--3072, Nov 2001.
\newblock \doi{10.1109/18.959288}.

\bibitem[Assmus and Key(1992)]{Assmus1992}
Edward~F. Assmus and Jennifer~D. Key.
\newblock \emph{Designs and {their} Codes}.
\newblock Cambridge Univ. Press, 1992.
\newblock \doi{10.1017/cbo9781316529836}.

\bibitem[Bosma et~al.(1997)Bosma, Cannon, and Playoust]{Bosma1997}
Wieb Bosma, John Cannon, and Catherine Playoust.
\newblock The {Magma} algebra system {I: The} user language.
\newblock \emph{Journal of Symbolic Computation}, 24\penalty0 (3-4):\penalty0
  235--265, Sep 1997.
\newblock \doi{10.1006/jsco.1996.0125}.

\bibitem[Bowen(2002)]{Bowen2002}
Garry Bowen.
\newblock Entanglement required in achieving entanglement-assisted channel
  capacities.
\newblock \emph{Physical Review A}, 66\penalty0 (5):\penalty0 052313, Nov 2002.
\newblock \doi{10.1103/physreva.66.052313}.

\bibitem[Brun et~al.(2006)Brun, Devetak, and Hsieh]{Brun2006}
Todd Brun, Igor Devetak, and Min-Hsiu Hsieh.
\newblock Correcting quantum errors with entanglement.
\newblock \emph{Science}, 314\penalty0 (5798):\penalty0 436--439, Oct 2006.
\newblock \doi{10.1126/science.1131563}.

\bibitem[Brun and Hsieh(2013)]{BrunHsieh}
Todd~A. Brun and Min-Hsiu Hsieh.
\newblock Entanglement-assisted quantum error-correcting codes.
\newblock In Daniel~A. Lidar and Todd~A. Brun, editors, \emph{Quantum Error
  Correction}, pages 181--200. Cambridge Univ. Press, 2013.
\newblock \doi{10.1017/cbo9781139034807.009}.

\bibitem[Brun et~al.(2014)Brun, Devetak, and Hsieh]{Brun2014}
Todd~A. Brun, Igor Devetak, and Min-Hsiu Hsieh.
\newblock Catalytic quantum error correction.
\newblock \emph{{IEEE} Transactions on Information Theory}, 60\penalty0
  (6):\penalty0 3073--3089, Jun 2014.
\newblock \doi{10.1109/tit.2014.2313559}.

\bibitem[Calderbank et~al.(1998)Calderbank, Rains, Shor, and
  Sloane]{Calderbank1998}
A.R. Calderbank, E.M. Rains, P.M. Shor, and N.J.A. Sloane.
\newblock Quantum error correction via codes over {GF(4)}.
\newblock \emph{{IEEE} Transactions on Information Theory}, 44\penalty0
  (4):\penalty0 1369--1387, Jul 1998.
\newblock \doi{10.1109/18.681315}.

\bibitem[Carlet et~al.(2018)Carlet, Mesnager, Tang, Qi, and
  Pellikaan]{Carlet2018}
Claude Carlet, Sihem Mesnager, Chunming Tang, Yanfeng Qi, and Ruud Pellikaan.
\newblock Linear codes over $\mathbb{F}_q$ are equivalent to {LCD} codes for $q
  > 3$.
\newblock \emph{{IEEE} Transactions on Information Theory}, 64\penalty0
  (4):\penalty0 3010--3017, Apr 2018.
\newblock \doi{10.1109/tit.2018.2789347}.

\bibitem[Chen(2022)]{Chen2022}
Hao Chen.
\newblock New {MDS} entanglement-assisted quantum codes from {MDS Hermitian}
  self-orthogonal codes.
\newblock \emph{arXiv preprint arXiv:2206.13995}, 2022.

\bibitem[Chen et~al.(2021{\natexlab{a}})Chen, Zhu, and Jiang]{Chen2021a}
Xiaojing Chen, Shixin Zhu, and Wan Jiang.
\newblock Cyclic codes and some new entanglement-assisted quantum {MDS} codes.
\newblock \emph{Designs, Codes and Cryptography}, 89\penalty0 (11):\penalty0
  2533--2551, Sep 2021{\natexlab{a}}.
\newblock \doi{10.1007/s10623-021-00935-y}.

\bibitem[Chen et~al.(2021{\natexlab{b}})Chen, Zhu, Jiang, and Luo]{Chen2021}
Xiaojing Chen, Shixin Zhu, Wan Jiang, and Gaojun Luo.
\newblock A new family of {EAQMDS} codes constructed from constacyclic codes.
\newblock \emph{Designs, Codes and Cryptography}, 89\penalty0 (9):\penalty0
  2179--2193, Jul 2021{\natexlab{b}}.
\newblock \doi{10.1007/s10623-021-00908-1}.

\bibitem[Devetak et~al.(2008)Devetak, Harrow, and Winter]{Devetak2008}
Igor Devetak, Aram~W. Harrow, and Andreas~J. Winter.
\newblock A resource framework for quantum {Shannon} theory.
\newblock \emph{{IEEE} Transactions on Information Theory}, 54\penalty0
  (10):\penalty0 4587--4618, Oct 2008.
\newblock \doi{10.1109/tit.2008.928980}.

\bibitem[Einstein et~al.(1935)Einstein, Podolsky, and Rosen]{Einstein1935}
Albert Einstein, Boris Podolsky, and Nathan Rosen.
\newblock Can quantum-mechanical description of physical reality be considered
  complete?
\newblock \emph{Physical Review}, 47\penalty0 (10):\penalty0 777--780, May
  1935.
\newblock \doi{10.1103/physrev.47.777}.

\bibitem[Ezerman et~al.(2019)Ezerman, Ling, Özkaya, and Solé]{ELOS19}
Martianus~Frederic Ezerman, San Ling, Buket Özkaya, and Patrick Solé.
\newblock Good stabilizer codes from quasi-cyclic codes over $\mathbb{F}_4$ and
  $\mathbb{F}_9$.
\newblock In \emph{2019 IEEE International Symposium on Information Theory
  (ISIT)}, pages 2898--2902, 2019.

\bibitem[Fang et~al.(2020)Fang, Fu, Li, and Zhu]{Fang2020}
Weijun Fang, Fang-Wei Fu, Lanqiang Li, and Shixin Zhu.
\newblock Euclidean and {Hermitian} hulls of {MDS} codes and their applications
  to {EAQECCs}.
\newblock \emph{{IEEE} Transactions on Information Theory}, 66\penalty0
  (6):\penalty0 3527--3537, Jun 2020.
\newblock \doi{10.1109/tit.2019.2950245}.

\bibitem[Galindo et~al.(2019)Galindo, Hernando, Matsumoto, and
  Ruano]{Galindo2019}
Carlos Galindo, Fernando Hernando, Ryutaroh Matsumoto, and Diego Ruano.
\newblock Entanglement-assisted quantum error-correcting codes over arbitrary
  finite fields.
\newblock \emph{Quantum Information Processing}, 18\penalty0 (4), Mar 2019.
\newblock \doi{10.1007/s11128-019-2234-5}.

\bibitem[Galindo et~al.(2021)Galindo, Hernando, Matsumoto, and
  Ruano]{Galindo2021}
Carlos Galindo, Fernando Hernando, Ryutaroh Matsumoto, and Diego Ruano.
\newblock Correction to: Entanglement-assisted quantum error-correcting codes
  over arbitrary finite fields.
\newblock \emph{Quantum Information Processing}, 20\penalty0 (6), Jun 2021.
\newblock \doi{10.1007/s11128-021-03066-4}.

\bibitem[Gottesman(1997)]{Gottesman1997}
Daniel Gottesman.
\newblock \emph{Stabilizer codes and quantum error correction}.
\newblock PhD thesis, California Institute of Technology, 1997.

\bibitem[Grassl(2007)]{Grassl:codetables}
Markus Grassl.
\newblock {Bounds on the minimum distance of linear codes and quantum codes}.
\newblock Online available at \url{http://www.codetables.de}, 2007.
\newblock Accessed on 2022-08-14.

\bibitem[Grassl(2022)]{Grassl:EAQECCtables}
Markus Grassl.
\newblock {Bounds on the minimum distance of entanglement-assisted quantum
  codes}.
\newblock Online available at \url{http://codetables.de/EAQECC/}, 2022.
\newblock Accessed on 2022-08-18.

\bibitem[Grassl et~al.(2022)Grassl, Huber, and Winter]{GraHubWin2022}
Markus Grassl, Felix Huber, and Andreas Winter.
\newblock Entropic proofs of {S}ingleton bounds for quantum error-correcting
  codes.
\newblock \emph{IEEE Transactions on Information Theory}, 68\penalty0
  (6):\penalty0 3942--3950, Jun 2022.
\newblock \doi{10.1109/TIT.2022.3149291}.

\bibitem[Griesmer(1960)]{Griesmer1960}
James~H. Griesmer.
\newblock A bound for error-correcting codes.
\newblock \emph{{IBM} Journal of Research and Development}, 4\penalty0
  (5):\penalty0 532--542, Nov 1960.
\newblock \doi{10.1147/rd.45.0532}.

\bibitem[Guenda et~al.(2017)Guenda, Jitman, and Gulliver]{Guenda2017}
Kenza Guenda, Somphong Jitman, and T.~Aaron Gulliver.
\newblock Constructions of good entanglement-assisted quantum error correcting
  codes.
\newblock \emph{Designs, Codes and Cryptography}, 86\penalty0 (1):\penalty0
  121--136, Jan 2017.
\newblock \doi{10.1007/s10623-017-0330-z}.

\bibitem[Horodecki et~al.(2009)Horodecki, Horodecki, Horodecki, and
  Horodecki]{Horodecki2009}
Ryszard Horodecki, Pawe{\l} Horodecki, Micha{\l} Horodecki, and Karol
  Horodecki.
\newblock Quantum entanglement.
\newblock \emph{Reviews of Modern Physics}, 81\penalty0 (2):\penalty0 865--942,
  Jun 2009.
\newblock \doi{10.1103/revmodphys.81.865}.

\bibitem[Hou(2018)]{Hou}
Xiang-Dong Hou.
\newblock \emph{Lectures on Finite Fields}.
\newblock Graduate Studies in Mathematics. American Mathematical Society, 2018.

\bibitem[Huffman(2005)]{Huffman2005}
W.~Cary Huffman.
\newblock On the classification and enumeration of self-dual codes.
\newblock \emph{Finite Fields and Their Applications}, 11\penalty0
  (3):\penalty0 451--490, Aug 2005.
\newblock \doi{10.1016/j.ffa.2005.05.012}.

\bibitem[Huffman et~al.(2021)Huffman, Kim, and Sol{\'{e}}]{Huffman2021}
W.~Cary Huffman, Jon-Lark Kim, and Patrick Sol{\'{e}}, editors.
\newblock \emph{Concise Encyclopedia of Coding Theory}.
\newblock Chapman and Hall/{CRC}, Mar 2021.
\newblock \doi{10.1201/9781315147901}.

\bibitem[Humphreys et~al.(2018)Humphreys, Kalb, Morits, Schouten, Vermeulen,
  Twitchen, Markham, and Hanson]{Humphreys2018}
Peter~C. Humphreys, Norbert Kalb, Jaco P.~J. Morits, Raymond~N. Schouten,
  Raymond F.~L. Vermeulen, Daniel~J. Twitchen, Matthew Markham, and Ronald
  Hanson.
\newblock Deterministic delivery of remote entanglement on a quantum network.
\newblock \emph{Nature}, 558\penalty0 (7709):\penalty0 268--273, Jun 2018.
\newblock \doi{10.1038/s41586-018-0200-5}.

\bibitem[Ketkar et~al.(2006)Ketkar, Klappenecker, Kumar, and
  Sarvepalli]{Ketkar2006}
Avanti Ketkar, Andreas Klappenecker, Santosh Kumar, and Pradeep~Kiran
  Sarvepalli.
\newblock Nonbinary stabilizer codes over finite fields.
\newblock \emph{{IEEE} Transactions on Information Theory}, 52\penalty0
  (11):\penalty0 4892--4914, Nov 2006.
\newblock \doi{10.1109/tit.2006.883612}.

\bibitem[Li et~al.(2015)Li, Li, and Guo]{Li2015}
Ruihu Li, Xueliang Li, and Luobin Guo.
\newblock On entanglement-assisted quantum codes achieving the
  entanglement-assisted {G}riesmer bound.
\newblock \emph{Quantum Information Processing}, 14\penalty0 (12):\penalty0
  4427--4447, Oct 2015.
\newblock \doi{10.1007/s11128-015-1143-5}.

\bibitem[Lison{\v{e}}k and Singh(2014)]{lisonvek2014quantum}
Petr Lison{\v{e}}k and Vijaykumar Singh.
\newblock Quantum codes from nearly self-orthogonal quaternary linear codes.
\newblock \emph{Designs, Codes and Cryptography}, 73\penalty0 (2):\penalty0
  417--424, 2014.

\bibitem[Luo et~al.(2019)Luo, Cao, and Chen]{Luo2019}
Gaojun Luo, Xiwang Cao, and Xiaojing Chen.
\newblock {MDS} codes with hulls of arbitrary dimensions and their quantum
  error correction.
\newblock \emph{{IEEE} Transactions on Information Theory}, 65\penalty0
  (5):\penalty0 2944--2952, May 2019.
\newblock \doi{10.1109/tit.2018.2874953}.

\bibitem[Luo et~al.(2022)Luo, Ezerman, and Ling]{Luo2022}
Gaojun Luo, Martianus~Frederic Ezerman, and San Ling.
\newblock Entanglement-assisted and subsystem quantum codes: New propagation
  rules and constructions.
\newblock \emph{arXiv preprint arXiv:2206.09782}, 2022.

\bibitem[Rains(1999)]{Rains99}
Eric~M. Rains.
\newblock {Nonbinary quantum codes}.
\newblock \emph{IEEE Transactions on Information Theory}, 45\penalty0
  (6):\penalty0 1827--1832, Sep 1999.
\newblock \doi{10.1109/18.782103}.

\bibitem[Wilde et~al.(2014)Wilde, Hsieh, and Babar]{Wilde2014}
Mark~M. Wilde, Min-Hsiu Hsieh, and Zunaira Babar.
\newblock Entanglement-assisted quantum turbo codes.
\newblock \emph{{IEEE} Transactions on Information Theory}, 60\penalty0
  (2):\penalty0 1203--1222, Feb 2014.
\newblock \doi{10.1109/tit.2013.2292052}.

\bibitem[Zhang et~al.(2021)Zhang, Tang, Zhou, and Ma]{Zhang2021}
Yihong Zhang, Yifan Tang, You Zhou, and Xiongfeng Ma.
\newblock Efficient entanglement generation and detection of generalized
  stabilizer states.
\newblock \emph{Physical Review A}, 103\penalty0 (5):\penalty0 052426, May
  2021.
\newblock \doi{10.1103/physreva.103.052426}.

\end{thebibliography}

\widetext

\begin{table*}
\caption{A Concise Version of the Parameters of Good Qubit EAQECCs with $3 \leq n \leq 64$. To obtain the full table, one applies the propagation rules in Section \ref{sec:compute}. We exclude the entries with $c=0$ since a large database for such codes is already available in \cite{Grassl:codetables}.}
\label{table:qubit}
\centering
\resizebox{0.9\textwidth}{!}{
\begin{tabular}{lll lll}
\hline
$\dsb{n,\kappa,\delta;c}_2 $ & $\dsb{n,\kappa,\delta;c}_2 $ & $\dsb{n,\kappa,\delta;c}_2 $ & $\dsb{n,\kappa,\delta;c}_2 $ & $\dsb{n,\kappa,\delta;c}_2 $ & $\dsb{n,\kappa,\delta;c}_2 $ \\
\hline
$ \dsb{3, 1, 3; 2}_2$ &
$ \dsb{5, 4, 2; 1}_2$ &
$ \dsb{5, 0, 4; 1}_2$ &
$ \dsb{6, 0, 6; 4}_2$ &
$ \dsb{7, 6, 2; 1}_2$ &
$ \dsb{7, 4, 3; 3}_2$ \\

$ \dsb{8, 5, 3; 3}_2$ &
$ \dsb{8, 4, 4; 4}_2$ &
$ \dsb{8, 2, 5; 4}_2$ &
$ \dsb{8, 0, 8; 6}_2$ &
$ \dsb{9, 8, 2; 1}_2$ &
$ \dsb{9, 6, 3; 3}_2$ \\

$ \dsb{9, 4, 5; 5}_2$ &
			$ \dsb{9, 2, 6; 5}_2$ &
			$ \dsb{10, 6, 4; 4}_2$ &
			$ \dsb{10, 4, 5; 4}_2$ &
			$ \dsb{10, 4, 6; 6}_2$ &
			$ \dsb{10, 0, 10; 8}_2$ \\
			
			$ \dsb{11, 10, 2; 1}_2$ &
			$ \dsb{11, 7, 3; 2}_2$ &
			$ \dsb{11, 6, 4; 3}_2$ &
			$ \dsb{11, 5, 6; 6}_2$ &
			$ \dsb{12, 9, 3; 3}_2$ &
			$ \dsb{12, 8, 4; 4}_2$ \\
			
			$ \dsb{12, 5, 6; 5}_2$ &
			$ \dsb{12, 0, 12; 10}_2$ &
			$ \dsb{13, 12, 2; 1}_2$ &
			$ \dsb{13, 10, 3; 3}_2$ &
			$ \dsb{13, 9, 4; 4}_2$ &
			$ \dsb{14, 0, 14; 12}_2$ \\
			
			$ \dsb{15, 14, 2; 1}_2$ &
			$ \dsb{15, 9, 5; 6}_2$ &
			$ \dsb{16, 13, 3; 3}_2$ &
			$ \dsb{16, 9, 4; 1}_2$ &
			$ \dsb{16, 9, 6; 7}_2$ &
			$ \dsb{16, 0, 16; 14}_2$ \\
			
			$ \dsb{17, 16, 2; 1}_2$ &
			$ \dsb{17, 13, 3; 2}_2$ &
			$ \dsb{17, 9, 6; 6}_2$ &
			$ \dsb{17, 0, 8; 1}_2$ &
			$ \dsb{17, 0, 12; 9}_2$ &
			$ \dsb{18, 15, 3; 3}_2$ \\
			
			$ \dsb{18, 10, 5; 4}_2$ &
			$ \dsb{18, 7, 9; 11}_2$ &
			$ \dsb{18, 0, 10; 6}_2$ &
			$ \dsb{18, 0, 18; 16}_2$ &
			$ \dsb{19, 18, 2; 1}_2$ &
			$ \dsb{19, 13, 4; 4}_2$ \\
			
			$ \dsb{20, 15, 3; 1}_2$ &
			$ \dsb{20, 15, 4; 5}_2$ &
			$ \dsb{20, 14, 5; 6}_2$ &
			$ \dsb{20, 10, 6; 4}_2$ &
			$ \dsb{20, 9, 8; 9}_2$ &
			$ \dsb{20, 0, 20; 18}_2$ \\
			
			$ \dsb{21, 20, 2; 1}_2$ &
			$ \dsb{21, 16, 4; 5}_2$ &
			$ \dsb{21, 15, 5; 6}_2$ &
			$ \dsb{21, 9, 7; 6}_2$ &
			$ \dsb{22, 16, 4; 4}_2$ &
			$ \dsb{22, 12, 5; 4}_2$ \\
			
			$ \dsb{22, 0, 22; 20}_2$ &
			$ \dsb{23, 22, 2; 1}_2$ &
			$ \dsb{23, 18, 4; 5}_2$ &
			$ \dsb{23, 14, 5; 5}_2$ &
			$ \dsb{23, 14, 6; 7}_2$ &
			$ \dsb{24, 18, 4; 4}_2$ \\
			
			$ \dsb{24, 16, 5; 6}_2$ &
			$ \dsb{24, 16, 6; 8}_2$ &
			$ \dsb{24, 0, 24; 22}_2$ &
			$ \dsb{25, 24, 2; 1}_2$ &
			$ \dsb{25, 20, 3; 3}_2$ &
			$ \dsb{25, 19, 4; 4}_2$ \\
			
			$ \dsb{25, 11, 7; 6}_2$ &
			$ \dsb{25, 0, 14; 11}_2$ &
			$ \dsb{26, 22, 3; 4}_2$ &
			$ \dsb{26, 21, 4; 5}_2$ &
			$ \dsb{26, 13, 7; 7}_2$ &
			$ \dsb{26, 0, 26; 24}_2$ \\
			
			$ \dsb{27, 26, 2; 1}_2$ &
			$ \dsb{27, 23, 3; 4}_2$ &
			$ \dsb{27, 22, 4; 5}_2$ &
			$ \dsb{28, 23, 3; 3}_2$ &
			$ \dsb{28, 23, 4; 5}_2$ &
			$ \dsb{28, 0, 28; 26}_2$ \\
			
			$ \dsb{29, 28, 2; 1}_2$ &
			$ \dsb{29, 25, 3; 4}_2$ &
			$ \dsb{29, 23, 4; 4}_2$ &
			$ \dsb{29, 0, 12; 1}_2$ &
			$ \dsb{30, 26, 3; 4}_2$ &
			$ \dsb{30, 25, 4; 5}_2$ \\
			
			$ \dsb{30, 0, 16; 12}_2$ &
			$ \dsb{30, 0, 30; 28}_2$ &
			$ \dsb{31, 30, 2; 1}_2$ &
			$ \dsb{31, 26, 4; 5}_2$ &
			$ \dsb{31, 22, 5; 5}_2$ &
			$ \dsb{31, 16, 6; 1}_2$ \\
			
			$ \dsb{31, 16, 9; 13}_2$ &
			$ \dsb{31, 9, 13; 16}_2$ &
			$ \dsb{31, 2, 14; 11}_2$ &
			$ \dsb{32, 27, 3; 3}_2$ &
			$ \dsb{32, 26, 4; 4}_2$ &
			$ \dsb{32, 24, 5; 6}_2$ \\
			
			$ \dsb{32, 0, 32; 30}_2$ &
			$ \dsb{33, 32, 2; 1}_2$ &
			$ \dsb{33, 28, 4; 5}_2$ &
			$ \dsb{33, 26, 5; 7}_2$ &
			$ \dsb{34, 29, 3; 3}_2$ &
			$ \dsb{34, 28, 4; 4}_2$ \\
			
			$ \dsb{34, 26, 5; 6}_2$ &
			$ \dsb{34, 24, 6; 8}_2$ &
			$ \dsb{34, 0, 14; 8}_2$ &
			$ \dsb{34, 0, 34; 32}_2$ &
			$ \dsb{35, 34, 2; 1}_2$ &
			$ \dsb{35, 30, 4; 5}_2$ \\
			
			$ \dsb{35, 28, 5; 7}_2$ &
			$ \dsb{35, 26, 6; 9}_2$ &
			$ \dsb{35, 12, 8; 1}_2$ &
			$ \dsb{35, 4, 14; 11}_2$ &
			$ \dsb{36, 32, 3; 4}_2$ &
			$ \dsb{36, 30, 4; 4}_2$ \\
			
			$ \dsb{36, 28, 5; 6}_2$ &
			$ \dsb{36, 26, 6; 8}_2$ &
			$ \dsb{36, 17, 10; 13}_2$ &
			$ \dsb{36, 9, 17; 23}_2$ &
			$ \dsb{36, 0, 36; 34}_2$ &
			$ \dsb{37, 36, 2; 1}_2$ \\
			
			$ \dsb{37, 30, 5; 7}_2$ &
			$ \dsb{37, 28, 6; 9}_2$ &
			$ \dsb{37, 14, 13; 19}_2$ &
			$ \dsb{38, 28, 6; 8}_2$ &
			$ \dsb{38, 16, 13; 20}_2$ &
			$ \dsb{38, 16, 14; 22}_2$ \\
			
			$ \dsb{38, 14, 8; 2}_2$ &
			$ \dsb{38, 0, 38; 36}_2$ &
			$ \dsb{39, 38, 2; 1}_2$ &
			$ \dsb{39, 30, 6; 9}_2$ &
			$ \dsb{39, 22, 10; 17}_2$ &
			$ \dsb{39, 16, 8; 3}_2$ \\
			
			$ \dsb{39, 15, 11; 12}_2$ &
			$ \dsb{39, 13, 17; 26}_2$ &
			$ \dsb{39, 12, 9; 3}_2$ &
			$ \dsb{39, 0, 18; 15}_2$ &
			$ \dsb{40, 36, 3; 4}_2$ &
			$ \dsb{40, 31, 4; 1}_2$ \\

			$ \dsb{40, 30, 5; 4}_2$ &
			$ \dsb{40, 30, 6; 8}_2$ &
			$ \dsb{40, 22, 10; 16}_2$ &
			$ \dsb{40, 20, 12; 20}_2$ &
			$ \dsb{40, 18, 13; 20}_2$ &
			$ \dsb{40, 16, 14; 20}_2$ \\
			
			$ \dsb{40, 0, 40; 38}_2$ &
			$ \dsb{41, 40, 2; 1}_2$ &
			$ \dsb{41, 36, 3; 3}_2$ &
			$ \dsb{41, 32, 5; 5}_2$ &
			$ \dsb{41, 32, 6; 9}_2$ &
			$ \dsb{41, 20, 7; 1}_2$ \\
			
			$ \dsb{41, 20, 13; 21}_2$ &
			$ \dsb{41, 18, 14; 21}_2$ &
			$ \dsb{42, 38, 3; 4}_2$ &
			$ \dsb{42, 34, 5; 6}_2$ &
			$ \dsb{42, 33, 6; 9}_2$ &
			$ \dsb{42, 28, 8; 14}_2$ \\
			
			$ \dsb{42, 22, 12; 20}_2$ &
			$ \dsb{42, 20, 11; 16}_2$ &
			$ \dsb{42, 20, 13; 20}_2$ &
			$ \dsb{42, 20, 14; 22}_2$ &
			$ \dsb{42, 16, 9; 6}_2$ &
			$ \dsb{42, 0, 16; 8}_2$ \\
			
			$ \dsb{42, 0, 42; 40}_2$ &
			$ \dsb{43, 42, 2; 1}_2$ &
			$ \dsb{43, 39, 3; 4}_2$ &
			$ \dsb{43, 36, 5; 7}_2$ &
			$ \dsb{43, 34, 4; 3}_2$ &
			$ \dsb{43, 34, 6; 9}_2$ \\
			
			$ \dsb{43, 31, 6; 8}_2$ &
			$ \dsb{43, 29, 8; 14}_2$ &
			$ \dsb{43, 21, 14; 22}_2$ &
			$ \dsb{43, 0, 18; 13}_2$ &
			$ \dsb{44, 39, 3; 3}_2$ &
			$ \dsb{44, 36, 4; 4}_2$ \\
			
			$ \dsb{44, 29, 8; 13}_2$ &
			$ \dsb{44, 21, 14; 21}_2$ &
			$ \dsb{44, 0, 44; 42}_2$ &
			$ \dsb{45, 44, 2; 1}_2$ &
			$ \dsb{45, 41, 3; 4}_2$ &
			$ \dsb{45, 31, 8; 14}_2$ \\
			
			$ \dsb{45, 24, 12; 21}_2$ &
			$ \dsb{45, 22, 7; 1}_2$ &
			$ \dsb{46, 41, 3; 3}_2$ &
			$ \dsb{46, 38, 4; 4}_2$ &
			$ \dsb{46, 38, 5; 8}_2$ &
			$ \dsb{46, 34, 6; 6}_2$ \\
			
			$ \dsb{46, 28, 8; 12}_2$ &
			$ \dsb{46, 32, 8; 14}_2$ &
			$ \dsb{46, 19, 9; 5}_2$ &
			$ \dsb{46, 17, 10; 7}_2$ &
			$ \dsb{46, 25, 11; 19}_2$ &
			$ \dsb{46, 24, 12; 20}_2$ \\
			
			$ \dsb{46, 1, 15; 7}_2$ &
			$ \dsb{46, 1, 19; 17}_2$ &
			$ \dsb{46, 0, 46; 44}_2$ &
			$ \dsb{47, 46, 2; 1}_2$ &
			$ \dsb{47, 41, 3; 2}_2$ &
			$ \dsb{47, 40, 4; 5}_2$ \\
			
			$ \dsb{47, 33, 8; 14}_2$ &
			$ \dsb{48, 43, 3; 3}_2$ &
			$ \dsb{48, 42, 4; 6}_2$ &
			$ \dsb{48, 35, 7; 13}_2$ &
			$ \dsb{48, 33, 8; 13}_2$ &
			$ \dsb{48, 23, 13; 21}_2$ \\
			
			$ \dsb{48, 16, 17; 26}_2$ &
			$ \dsb{48, 0, 20; 16}_2$ &
			$ \dsb{48, 0, 48; 46}_2$ &
			$ \dsb{49, 48, 2; 1}_2$ &
			$ \dsb{49, 45, 3; 4}_2$ &
			$ \dsb{49, 43, 4; 6}_2$ \\
			
			$ \dsb{49, 35, 8; 14}_2$ &
			$ \dsb{50, 46, 3; 4}_2$ &
			$ \dsb{50, 43, 4; 5}_2$ &
			$ \dsb{50, 36, 7; 12}_2$ &
			$ \dsb{50, 35, 8; 13}_2$ &
			$ \dsb{50, 19, 9; 1}_2$ \\
			
			$ \dsb{50, 0, 50; 48}_2$ &
			$ \dsb{51, 50, 2; 1}_2$ &
			$ \dsb{51, 46, 3; 3}_2$ &
			$ \dsb{51, 45, 4; 6}_2$ &
			$ \dsb{51, 43, 5; 8}_2$ &
			$ \dsb{51, 34, 6; 1}_2$ \\
			
			$ \dsb{51, 38, 7; 13}_2$ &
			$ \dsb{51, 37, 8; 14}_2$ &
			$ \dsb{51, 9, 12; 2}_2$ &
			$ \dsb{52, 48, 3; 4}_2$ &
			$ \dsb{52, 46, 4; 6}_2$ &
			$ \dsb{52, 44, 5; 8}_2$ \\
			
			$ \dsb{52, 17, 10; 1}_2$ &
			$ \dsb{52, 0, 22; 18}_2$ &
			$ \dsb{52, 0, 52; 50}_2$ &
			$ \dsb{53, 52, 2; 1}_2$ &
			$ \dsb{53, 48, 3; 3}_2$ &
			$ \dsb{53, 47, 4; 6}_2$ \\
			
			$ \dsb{53, 44, 5; 7}_2$ &
			$ \dsb{53, 19, 10; 2}_2$ &
			$ \dsb{53, 0, 16; 1}_2$ &
			$ \dsb{54, 50, 3; 4}_2$ &
			$ \dsb{54, 46, 5; 8}_2$ &
			$ \dsb{54, 44, 6; 10}_2$ \\
			
			$ \dsb{54, 0, 54; 52}_2$ &
			$ \dsb{55, 54, 2; 1}_2$ &
			$ \dsb{55, 51, 3; 4}_2$ &
			$ \dsb{55, 47, 4; 4}_2$ &
			$ \dsb{55, 47, 5; 8}_2$ &
			$ \dsb{56, 52, 3; 4}_2$ \\
			
			$ \dsb{56, 48, 5; 8}_2$ &
			$ \dsb{56, 0, 24; 20}_2$ &
			$ \dsb{56, 0, 56; 54}_2$ &
			$ \dsb{57, 56, 2; 1}_2$ &
			$ \dsb{57, 52, 3; 3}_2$ &
			$ \dsb{57, 49, 4; 4}_2$ \\
			
			$ \dsb{57, 49, 5; 8}_2$ &
			$ \dsb{58, 54, 3; 4}_2$ &
			$ \dsb{58, 0, 58; 56}_2$ &
			$ \dsb{59, 58, 2; 1}_2$ &
			$ \dsb{59, 55, 3; 4}_2$ &
			$ \dsb{60, 55, 3; 3}_2$ \\

			$ \dsb{60, 0, 60; 58}_2$ &
			$ \dsb{61, 60, 2; 1}_2$ &
			$ \dsb{61, 57, 3; 4}_2$ &
			$ \dsb{61, 0, 18; 1}_2$ &
			$ \dsb{62, 58, 3; 4}_2$ &
			$ \dsb{62, 51, 4; 1}_2$ \\
			
			$ \dsb{62, 0, 62; 60}_2$ &
			$ \dsb{63, 62, 2; 1}_2$ &
			$ \dsb{63, 58, 3; 3}_2$ &
			$ \dsb{63, 44, 6; 1}_2$ &
			$ \dsb{63, 17, 13; 4}_2$ &
			$ \dsb{63, 16, 14; 5}_2$ \\
			
			$ \dsb{64, 60, 3; 4}_2$ &
			$ \dsb{64, 53, 4; 1}_2$ &
			$ \dsb{64, 49, 5; 1}_2$ &
			$ \dsb{64, 40, 10; 14}_2$ &
			$ \dsb{64, 37, 8; 1}_2$ &
			$ \dsb{64, 31, 9; 1}_2$ \\
			
			$ \dsb{64, 25, 11; 1}_2$ &
			$ \dsb{64, 24, 12; 2}_2$ &
			$ \dsb{64, 21, 16; 17}_2$ &
			$ \dsb{64, 17, 15; 9}_2$ &
			$ \dsb{64, 16, 16; 14}_2$ & 
			$ \dsb{64, 12, 14; 2}_2$ \\
			
			$ \dsb{64, 2, 22; 12}_2$ &
			$ \dsb{64, 1, 20; 9}_2$ &
			$ \dsb{64, 1, 23; 13}_2$ &
			$ \dsb{64, 1, 27; 25}_2$ &
			$ \dsb{64, 0, 24; 14}_2$ &
			$ \dsb{64, 0, 64; 62}_2$ \\
\hline
\end{tabular}
}
\end{table*}

\pagebreak

\begin{table*}
\caption{A Concise Version of the Parameters of Good Qutrit EAQECCs with $3 \leq n \leq 36$. To obtain the full table, one applies the propagation rules in Section \ref{sec:compute}. We include entries with $c=0$ since a large database for such codes is not yet currently available online.}
\label{table:qutrit}
\centering
\resizebox{0.9\textwidth}{!}{
\begin{tabular}{lll lll}
\hline
$\dsb{n,\kappa,\delta;c}_3 $ & $\dsb{n,\kappa,\delta;c}_3 $ & $\dsb{n,\kappa,\delta;c}_3 $ & $\dsb{n,\kappa,\delta;c}_3 $ & $\dsb{n,\kappa,\delta;c}_3 $ & $\dsb{n,\kappa,\delta;c}_3 $ \\
\hline

$ \dsb{3, 0, 3; 1 }_3$ &
$ \dsb{5, 4, 2; 1 }_3$ &
		$ \dsb{5, 2, 3; 1 }_3$ &
		$ \dsb{5, 0, 4; 1 }_3$ &
		$ \dsb{5, 0, 5; 3 }_3$ &
		$ \dsb{6, 2, 4; 2 }_3$ \\
		
		$ \dsb{7, 3, 3; 0 }_3$ &
		$ \dsb{7, 0, 6; 3 }_3$ &
		$ \dsb{8, 4, 3; 0 }_3$ &
		$ \dsb{8, 0, 7; 4 }_3$ &
		$ \dsb{10, 6, 3; 0 }_3$ &
		$ \dsb{10, 4, 4; 0 }_3$ \\
		
		$ \dsb{10, 4, 5; 2 }_3$ &
		$ \dsb{10, 1, 6; 1 }_3$ &
		$ \dsb{10, 2, 7; 4 }_3$ &
		$ \dsb{10, 0, 8; 4 }_3$ &
		$ \dsb{10, 0, 9; 6 }_3$ &
		$ \dsb{10, 0, 10; 8 }_3$ \\
		
		$ \dsb{11, 0, 11; 9 }_3$ &
		$ \dsb{12, 0, 12; 10 }_3$ &
		$ \dsb{13, 4, 7; 5 }_3$ &
		$ \dsb{13, 1, 10; 8 }_3$ &
		$ \dsb{13, 0, 13; 11 }_3$ &
		$ \dsb{14, 8, 3; 0 }_3$ \\
		
		$ \dsb{14, 5, 6; 3 }_3$ &
		$ \dsb{14, 2, 9; 6 }_3$ &
		$ \dsb{14, 0, 14; 12 }_3$ &
		$ \dsb{15, 9, 3; 0 }_3$ &
		$ \dsb{15, 5, 5; 0 }_3$ &
		$ \dsb{15, 5, 7; 4 }_3$ \\
		
		$ \dsb{15, 4, 8; 5 }_3$ &
		$ \dsb{15, 0, 12; 9 }_3$ &
		$ \dsb{15, 0, 15; 13 }_3$ &
		$ \dsb{16, 10, 3; 0 }_3$ &
		$ \dsb{16, 9, 4; 1 }_3$ &
		$ \dsb{16, 7, 5; 1 }_3$ \\
		
		$ \dsb{16, 6, 6; 2 }_3$ &
		$ \dsb{16, 6, 7; 4 }_3$ &
		$ \dsb{16, 5, 8; 5 }_3$ &
		$ \dsb{16, 4, 9; 6 }_3$ &
		$ \dsb{16, 2, 10; 6 }_3$ &
		$ \dsb{16, 1, 11; 7 }_3$ \\
		
		$ \dsb{16, 1, 12; 9 }_3$ &
		$ \dsb{16, 0, 13; 10 }_3$ &
		$ \dsb{16, 0, 16; 14 }_3$ &
		$ \dsb{17, 11, 3; 0 }_3$ &
		$ \dsb{17, 10, 4; 1 }_3$ &
		$ \dsb{17, 9, 5; 2 }_3$ \\
		
		$ \dsb{17, 8, 6; 3 }_3$ &
		$ \dsb{17, 0, 14; 11 }_3$ &
		$ \dsb{17, 0, 17; 15 }_3$ &
		$ \dsb{18, 12, 3; 0 }_3$ &
		$ \dsb{18, 11, 4; 1 }_3$ &
		$ \dsb{18, 10, 5; 2 }_3$ \\
		
		$ \dsb{18, 6, 8; 4 }_3$ &
		$ \dsb{18, 6, 9; 6 }_3$ &
		$ \dsb{18, 4, 10; 6 }_3$ &
		$ \dsb{18, 0, 18; 16 }_3$ &
		$ \dsb{19, 13, 3; 0 }_3$ &
		$ \dsb{19, 12, 4; 1 }_3$ \\
		
		$ \dsb{19, 11, 5; 2 }_3$ &
		$ \dsb{19, 2, 13; 11 }_3$ &
		$ \dsb{19, 0, 19; 17 }_3$ &
		$ \dsb{19, 0, 20, 19 }_3$ &
		$ \dsb{20, 12, 5; 2 }_3$ &
		$ \dsb{20, 10, 6; 4 }_3$ \\
		
		$ \dsb{20, 9, 7; 5 }_3$ &
		$ \dsb{20, 6, 10; 6 }_3$ &
		$ \dsb{20, 5, 11; 9 }_3$ &
		$ \dsb{20, 3, 12; 9 }_3$ &
		$ \dsb{20, 0, 15; 12 }_3$ &
		$ \dsb{20, 0, 20; 18 }_3$ \\
		
		$ \dsb{21, 15, 3; 0 }_3$ &
		$ \dsb{21, 13, 4; 0 }_3$ &
		$ \dsb{21, 11, 6; 4 }_3$ &
		$ \dsb{21, 10, 7; 5 }_3$ &
		$ \dsb{21, 2, 14; 11 }_3$ &
		$ \dsb{21, 2, 15; 13 }_3$ \\
		
		$ \dsb{21, 1, 16; 14 }_3$ &
		$ \dsb{21, 0, 21; 19 }_3$ &
		$ \dsb{22, 16, 3; 0 }_3$ &
		$ \dsb{22, 15, 4; 1 }_3$ &
		$ \dsb{22, 12, 6; 4 }_3$ &
		$ \dsb{22, 11, 7; 5 }_3$ \\
		
		$ \dsb{22, 1, 17; 15 }_3$ &
		$ \dsb{22, 0, 22; 20 }_3$ &
		$ \dsb{23, 17, 3; 0 }_3$ &
		$ \dsb{23, 16, 4; 1 }_3$ &
		$ \dsb{23, 9, 8; 6 }_3$ &
		$ \dsb{23, 4, 13; 11 }_3$ \\
		
		$ \dsb{23, 0, 23; 21 }_3$ &
		$ \dsb{24, 13, 6; 3 }_3$ &
		$ \dsb{24, 11, 8; 7 }_3$ &
		$ \dsb{24, 6, 11; 8 }_3$ &
		$ \dsb{24, 6, 12; 10 }_3$ &
		$ \dsb{24, 0, 19; 16 }_3$ \\
		
		$ \dsb{24, 0, 24; 22 }_3$ &
		$ \dsb{25, 17, 4; 0 }_3$ &
		$ \dsb{25, 15, 5; 2 }_3$ &
		$ \dsb{25, 9, 10; 8 }_3$ &
		$ \dsb{25, 5, 14; 12 }_3$ &
		$ \dsb{25, 3, 16; 14 }_3$ \\
		
		$ \dsb{25, 2, 17; 15 }_3$ &
		$ \dsb{25, 0, 25; 23 }_3$ &
		$ \dsb{26, 18, 4; 0 }_3$ &
		$ \dsb{26, 13, 7; 5 }_3$ &
		$ \dsb{26, 12, 8; 6 }_3$ &
		$ \dsb{26, 11, 9; 7 }_3$ \\
		
		$ \dsb{26, 10, 10; 8 }_3$ &
		$ \dsb{26, 7, 12; 9 }_3$ &
		$ \dsb{26, 6, 14; 12 }_3$ &
		$ \dsb{26, 3, 17; 15 }_3$ &
		$ \dsb{26, 2, 18; 16 }_3$ &
		$ \dsb{26, 0, 26; 24 }_3$ \\
		
		$ \dsb{27, 19, 4; 0 }_3$ &
		$ \dsb{27, 17, 5; 2 }_3$ &
		$ \dsb{27, 15, 6; 2 }_3$ &
		$ \dsb{27, 14, 7; 5 }_3$ &
		$ \dsb{27, 13, 8; 6 }_3$ &
		$ \dsb{27, 12, 9; 7 }_3$ \\
		
		$ \dsb{27, 11, 10; 8 }_3$ &
		$ \dsb{27, 8, 12; 9 }_3$ &
		$ \dsb{27, 8, 13; 11 }_3$ &
		$ \dsb{27, 7, 14; 12 }_3$ &
		$ \dsb{27, 4, 16; 13 }_3$ &
		$ \dsb{27, 3, 18; 16 }_3$ \\
		
		$ \dsb{27, 2, 19; 17 }_3$ &
		$ \dsb{27, 1, 20; 18 }_3$ &
		$ \dsb{27, 0, 27, 25 }_3$ &
		$ \dsb{28, 20, 4; 0 }_3$ &
		$ \dsb{28, 14, 8; 6 }_3$ &
		$ \dsb{28, 13, 9; 7 }_3$ \\
		
		$ \dsb{28, 12, 10; 8 }_3$ &
		$ \dsb{28, 11, 11; 9 }_3$ &
		$ \dsb{28, 10, 13; 12 }_3$ &
		$ \dsb{28, 9, 12; 9 }_3$ &
		$ \dsb{28, 8, 14; 12 }_3$ &
		$ \dsb{28, 5, 16; 13 }_3$ \\
		
		$ \dsb{28, 4, 17; 14 }_3$ &
		$ \dsb{28, 4, 18; 16 }_3$ &
		$ \dsb{28, 4, 19; 18 }_3$ &
		$ \dsb{28, 2, 20; 18 }_3$ &
		$ \dsb{28, 1, 21; 19 }_3$ &
		$ \dsb{28, 0, 28; 26 }_3$ \\
		
		$ \dsb{29, 23, 3; 0 }_3$ &
		$ \dsb{29, 21, 4; 0 }_3$ &
		$ \dsb{29, 19, 5; 2 }_3$ &
		$ \dsb{29, 17, 6; 2 }_3$ &
		$ \dsb{29, 16, 7; 5 }_3$ &
		$ \dsb{29, 0, 29; 27 }_3$ \\
		
		$ \dsb{30, 23, 4; 1 }_3$ &
		$ \dsb{30, 21, 6; 5 }_3$ &
		$ \dsb{30, 18, 7; 6 }_3$ &
		$ \dsb{30, 10, 15; 16 }_3$ &
		$ \dsb{30, 0, 12; 0 }_3$ &
		$ \dsb{30, 0, 30; 28 }_3$ \\
		
		$ \dsb{31, 25, 3; 0 }_3$ &
		$ \dsb{31, 25, 4; 2 }_3$ &
		$ \dsb{31, 22, 5; 3 }_3$ &
		$ \dsb{31, 22, 6; 5 }_3$ &
		$ \dsb{31, 0, 31; 29 }_3$ &
		$ \dsb{32, 26, 4; 2 }_3$ \\
		
		$ \dsb{32, 23, 5; 3 }_3$ &
		$ \dsb{32, 23, 6; 5 }_3$ &
		$ \dsb{32, 0, 32; 30 }_3$ &
		$ \dsb{33, 27, 3; 0 }_3$ &
		$ \dsb{33, 27, 4; 2 }_3$ &
		$ \dsb{33, 25, 5; 4 }_3$ \\
		
		$ \dsb{33, 24, 6; 5 }_3$ &
		$ \dsb{33, 14, 10; 9 }_3$ &
		$ \dsb{33, 13, 13; 14 }_3$ &
		$ \dsb{33, 13, 14; 16 }_3$ &
		$ \dsb{33, 12, 15; 17 }_3$ &
		$ \dsb{33, 11, 16; 18 }_3$ \\
		
		$ \dsb{33, 10, 17; 19 }_3$ &
		$ \dsb{33, 8, 18; 19 }_3$ &
		$ \dsb{33, 8, 19; 21 }_3$ &
		$ \dsb{33, 7, 20; 22 }_3$ &
		$ \dsb{33, 0, 33; 31 }_3$ &
		$ \dsb{34, 28, 3; 0 }_3$ \\
		
		$ \dsb{34, 26, 5; 4 }_3$ &
		$ \dsb{34, 25, 6; 5 }_3$ &
		$ \dsb{34, 12, 12; 10 }_3$ &
		$ \dsb{34, 0, 22; 18 }_3$ &
		$ \dsb{34, 4, 23; 24 }_3$ &
		$ \dsb{34, 0, 34; 32 }_3$ \\
		
		$ \dsb{35, 30, 3; 1 }_3$ &
		$ \dsb{35, 28, 4; 1 }_3$ &
		$ \dsb{35, 26, 6; 5 }_3$ &
		$ \dsb{35, 15, 10; 8 }_3$ &
		$ \dsb{35, 0, 35; 33 }_3$ &
		$ \dsb{36, 31, 3; 1 }_3$ \\
		
		$ \dsb{36, 30, 4; 2 }_3$ &
		$ \dsb{36, 27, 5; 3 }_3$ &
		$ \dsb{36, 27, 6; 5 }_3$ &
		$ \dsb{36, 20, 7; 2 }_3$ &
		$ \dsb{36, 19, 10; 11 }_3$ &
		$ \dsb{36, 18, 8; 4 }_3$ \\
		
		$ \dsb{36, 17, 9; 5 }_3$ &
		$ \dsb{36, 15, 11; 9 }_3$ &
		$ \dsb{36, 15, 14; 17 }_3$ &
		$ \dsb{36, 14, 13; 14 }_3$ &
		$ \dsb{36, 14, 15; 18 }_3$ &
		$ \dsb{36, 13, 12; 9 }_3$ \\
		
		$ \dsb{36, 11, 16; 17 }_3$ &
		$ \dsb{36, 11, 18; 21 }_3$ &
		$ \dsb{36, 10, 17; 18 }_3$ &
		$ \dsb{36, 10, 19; 22 }_3$ &
		$ \dsb{36, 8, 20; 22 }_3$ &
		$ \dsb{36, 8, 21; 24 }_3$ \\
		
		$ \dsb{36, 0, 36; 34 }_3$ &
		&&&&\\
		\hline 
\end{tabular}
}
\end{table*}
\end{document}